\documentclass[10pt]{article}

\usepackage[utf8]{inputenc}
\usepackage{microtype}
\usepackage{mathtools,mathrsfs,amsthm,amsmath,amssymb}    
\usepackage[table,svgnames]{xcolor}

\usepackage[hyphens]{url}
\usepackage[colorlinks=true,linkcolor=black,citecolor=MidnightBlue,urlcolor=MidnightBlue,breaklinks=true]{hyperref}

\usepackage{fullpage}
\usepackage{xspace}
\usepackage{tikz}
\usepackage{mdframed}
\usepackage{framed}
\usepackage{enumitem}
\usepackage[anythingbreaks]{breakurl}

\usepackage{algorithm}
\usepackage{algorithmicx}
\usepackage[noend]{algpseudocode}
\usepackage{caption}
\usepackage{algpseudocode}
\usepackage{refcount}
\makeatletter
\newcommand{\algmargin}{\the\ALG@thistlm}
\makeatother

\newtheorem{theorem}{Theorem}[section]
\newtheorem{lemma}[theorem]{Lemma}
\newtheorem{proposition}[theorem]{Proposition}
\newtheorem{corollary}[theorem]{Corollary}
   
\newtheorem{claim}[theorem]{Claim}
\newtheorem{assumption}[theorem]{Assumption}
\newtheorem{definition}{Definition}[section]
\newtheorem{remark}{Remark}

\DeclareMathOperator*{\argmax}{arg\,max}
\newcommand{\eps}{\varepsilon}
\DeclareMathOperator{\poly}{poly\!}
\newcommand\numberthis{\addtocounter{equation}{1}\tag{\theequation}}

\newcommand{\parse}{\textsc{Parse}}
\newcommand{\Parse}{\parse}
\newcommand{\Next}{\textsc{Next}\xspace}
\newcommand{\Trans}{\textsc{TempTranscript}\xspace}

\newcommand{\GOOD}{\mathsf{GOOD}}
\newcommand{\GOODi}[1]{\GOOD^{\le #1}}
\newcommand{\Skip}{\mathsf{SkipCnt}}
\newcommand{\Prev}{\mathsf{Prev}}
\newcommand{\Chain}{\textit{Chain}}
\newcommand{\GChain}{\textit{GoodChain}}

\newcommand{\noerr}{*}
\newcommand{\KW}{KW}
\newcommand{\AND}{\wedge}
\newcommand{\OR}{\vee}

\newcommand{\RC}{\mathsf{RC}}

\begin{document}

\title{Optimal Short-Circuit Resilient Formulas}

\author{%
Mark Braverman\thanks{Department of Computer Science, Princeton University, USA. \texttt{mbraverm@cs.princeton.edu}.} 
\and Klim Efremenko\thanks{Computer Science Department, Ben-Gurion University, Israel. \texttt{klimefrem@gmail.com}}
\and Ran Gelles\thanks{Faculty of Engineering, Bar-Ilan University, Israel. \texttt{ran.gelles@biu.ac.il}.}
\and Michael A. Yitayew\thanks{Department of Computer Science, Princeton University, USA.}
}%

\date{}

\maketitle

\begin{abstract}
We consider fault-tolerant boolean formulas in which the output of a faulty gate
is short-circuited to one of the gate's inputs. A recent result by Kalai et al. [FOCS 2012]
converts any boolean formula into a resilient formula of polynomial size that works correctly  if less than $1/6$ of the gates (on every input-to-output path) are faulty.
We improve the result of Kalai et al., and show how to efficiently fortify any boolean formula against
a fraction of $1/5$ of short-circuit gates per path, with only a polynomial blowup in size. We additionally show that it is impossible to obtain formulas with higher resilience and sub-exponential  growth in size.

\medskip

Towards our results, we consider interactive coding schemes when noiseless feedback is present; these produce  resilient boolean formulas via a Karchmer-Wigderson relation.
We develop a coding scheme that resists corruptions in up to a fraction of $1/5$ of the transmissions \emph{in each direction of the interactive channel}. We further show that such a level of noise is maximal for coding schemes whose communication blowup is sub-exponential. Our coding scheme has taken a surprising inspiration from Blockchain technology.
\end{abstract}

\section{Introduction}

Kleitman, Leighton and Ma~\cite{KLM97} asked the following question: assume you wish to build a logic circuit $C$ from AND and OR gates; however, due to some confusion, some small number of AND gates were placed in the box of the OR gates (and vice versa), and there is no way to distinguish between the two types of gates just by looking at them. Can you construct a ``resilient'' logic circuit $C'$ that computes the same functionality as~$C$, even if some (small number) of the AND gates are replaced with OR gates (and vice versa)?

The above toy question is a special case of a more general type of noise (faulty gates) known as \emph{short-circuit} noise. In this model, a faulty gate ``short-circuits'' one of its input-legs to the output-leg. That is, the output of the gate is determined by the value of one of its input-legs. The specific input that is connected to the output is determined by an all-powerful adversary, possibly as a function of the input to the circuit. This model is equivalent to a setting in which a faulty gate can be replaced with an arbitrary function~$g$, as long as it holds that $g(0,0)=0$ and $g(1,1)=1$. 
Note that this type of noise is different from the so-called von Neumann noise model for circuits~\cite{vonN56}, in which the noise flips the value of each wire in the circuit independently with probability~$p$. See~\cite{KLM97,KLR12} and references therein for a comparison between these two separate models.

The first solution to the above question---constructing circuits that are resilient to short-circuit faults---was provided by Kleitman et al.~\cite{KLM97}. They showed that for any number~$e$, 
a circuit of size $|C|$ gates can be transformed into a ``resilient'' circuit of size $|C'|$ that behaves correctly even if up to $e$ of its gates are faulty (short-circuited), and it further holds that $|C'| \le O(e\cdot |C| + e^{\log 3})$. 

Further progress was made by Kalai, Lewko, and Rao~\cite{KLR12}; they showed, for any constant $\eps>0$, 
how to convert any formula\footnote{A formula is a circuit in which each gate has a fan-out of~1.} $F$ of size $|F|$ into a resilient formula $F'$ of size $|F'|=\poly_\eps(|F|)$ such that $F'$ computes the same function that $F$~computes, as long as at most a fraction of $\left(\tfrac1{6}-\eps\right)$ of the gates in \emph{any input-to-output path} in $F'$ suffer from short-circuit noise. 
Kalai et al.\@ explicitly leave open the question of finding the \emph{optimal} fraction of faulty gates for a resilient formula~$F'$.\footnote{For instance, it is clear  that if all the gates in an input-to-output path can be short-circuited (i.e., the fraction of noise is~1), then the adversary has full control of the output of the circuit. Hence, the optimal noise rate for %
formulas lies within the range $[\frac1{6},1]$.}

In this work we show that a fraction of~$\tfrac15$ is a tight bound on the tolerable fraction of faulty gates per input-to-output path, subject to the condition that the increase in the size of the formula is sub-exponential.
Namely, we show how to convert any formula to a resilient version that tolerates up to a fraction of~$\tfrac15-\eps$ of short-circuited gates per path.
\begin{theorem}[Main, informal]\label{thm:main}
For any $\eps>0$, any formula~$F$ can be efficiently converted into a formula $F'$ of size $|F'|=\poly_\eps(|F|)$ 
that computes the same function as $F$  even when up to $\tfrac15-\eps$ of the gates in any of its input-to-output paths are short-circuited. 
\end{theorem}

We also show that our bound is tight. 
Namely, for an arbitrary formula~$F$, it is impossible to make a resilient version (of sub-exponential size in~$|F|$)
that tolerates a fraction $\tfrac15$ (or more) of short-circuited gates per path.
\begin{theorem}[Converse]\label{thm:converse-inf}
There exists a formula~$F$ for computing some function~$f$,
such that no formula~$F'$ of size $|F'|=o(\exp(|F|))$ that computes~$f$
is resilient to a fraction of~$\tfrac15$ of short-circuit noise in any of its input-to-output paths.  
\end{theorem}

Similar to the work of Kalai et al.~\cite{KLR12}, a major ingredient in our result is a transformation, known as the Karchmer-Wigderson transformation (hereinafter, the KW-transformation)~\cite{KW90}, between a formula that computes a boolean function~$f$, and
a two-party interactive communication protocol for a task related to~$f$, which we denote the KW-game for $f$, or~$\KW_f$ for short. Similarly, a reverse KW-transformation converts protocols back to formulas; see below and Section~\ref{sec:KW} for more details on the KW-transformation.
The work of Kalai et al.\@ adapts the KW-transformation to a noisy setting in which the formula may suffer from short-circuit noise, and the protocol may suffer from channel noise.  
The ``attack plan'' in~\cite{KLR12} for making a given formula $F$ resilient to short-circuit noise is (i) apply the KW-transformation to obtain an interactive protocol~$\pi$; (ii) convert $\pi$ to a noise-resilient protocol~$\pi'$ that tolerates up to a $\delta$-fraction of noise; (iii) apply the (reverse) KW-transformation on $\pi'$ to obtain a formula~$F'$. The analysis of~\cite{KLR12} shows that the obtained~$F'$ is resilient to a~$\delta/2$ fraction of noise in any of its input-to-output paths. 

The interactive protocols~$\pi,\pi'$ are defined in a setting where the parties have access to a noiseless feedback channel---the sender learns whether or not its transmission arrived correctly at the other side. Building upon recent progress in the field of coding for interactive protocols (see, e.g.,~\cite{gelles17}), Kalai et al.~\cite{KLR12} constructed a coding scheme for interactive protocols (with noiseless feedback) that features resilience of $\delta=\tfrac13-\eps$ for any $\eps>0$; %
this gives their result. Note that a resilience of~$\delta=\tfrac{1}{3}$ is maximal for interactive protocols in that setting~\cite{EGH16}, which implies that new techniques must be introduced in order to improve the result by~\cite{KLR12}.

The loss in resilience witnessed in step (iii) stems from the fact that short-circuit noise affects formulas in a ``one-sided'' manner:
a short-circuit of an AND gate can only turn the output from $0$ to~$1$, while a short-circuit in an OR gate can only turn the output from $1$ to $0$. 
The noisy AND gates are thus decoupled from the noisy OR gates:
if the output of the circuit is $0$, any amount of short-circuited OR gates will keep the output~$0$,
while
if the output is~$1$, any amount of short-circuited AND gates will keep the output~$1$
(see Lemma~\ref{lem:one-sided-noise}).
Informally speaking, this decoupling reduces by half the resilience of circuits generated by the KW-transformation.
Assume the formula~$F'$ obtained from the above process
is resilient to a $\delta'$-fraction of noise. Then~$F'$ is correct if on a specific input-to-output path
(a) at most a $\delta'$-fraction of the AND gates are short-circuited, but also if
(b) at most a $\delta'$-fraction of the OR gates are short-circuited. Since the noise is decoupled, from (a) and (b) we get that $F$ 
outputs the correct value even when a $2\delta'$-fraction of the gates on that input-to-output path are noisy. 
Yet, the resilience of~$F'$ originates from the resilience of $\pi'$ (step~(iii) above).
The KW-transformation limits the resilience of~$F'$
by the resilience of~$\pi'$, i.e., $2\delta' \le \delta$,
leading to a factor~2 loss.\looseness=-1

We revisit the above line of thought and make a more careful noise analysis. 
Instead of bounding the \emph{total}
fraction of noise by some~$\delta$, we consider the case where  the noise from Alice to Bob is bounded by some~$\alpha$ while the noise in the other direction is bounded by some~$\beta$. A similar approach used by Braverman and Efremenko~\cite{BE17} yields interactive protocols (without  noiseless feedback) with maximal resilience.
In more detail, assume that the protocol~$\pi$ communicates~$n$ symbols overall. We define an $(\alpha,\beta)$-corruption as any noise that corrupts up to $\alpha n$ symbols sent by Alice and up to $\beta n$ symbols sent by Bob. We emphasize  that the noise \emph{fraction} on Alice's transmissions is higher than~$\alpha$, since Alice speaks less than $n$ symbols overall; the global noise fraction in this case is $\alpha+\beta$.

This distinction may be delicate but is instrumental. 
The KW-transformation translates 
a protocol of length~$n$ that is resilient to $(\alpha,\beta)$-corruptions
into a formula which is resilient to up to $\alpha n$ short-circuited AND gates \emph{in addition to} up to $\beta n$ short-circuited OR gates. When $\alpha=\beta$ the obtained formula is resilient to up to an $\alpha$-fraction of short-circuited gates in any input-to-output path, avoiding the factor 2 loss in resilience.

\goodbreak

\subsection{Technique overview}
\paragraph{Achievability: Coding schemes for noisy channels with noiseless feedback.}
We obtain resilient formulas by employing the approach of~\cite{KLR12} described above. 
In order to increase 
the noise resilience to its optimal level, we develop a novel coding scheme which is resilient to 
$\left(\tfrac{1}{5}-\eps,\tfrac{1}{5}-\eps\right)$-corruptions, assuming noiseless feedback.

The mechanism of our coding scheme resembles, in a sense, the \emph{Blockchain technology}~\cite{nak08}. Given a protocol $\pi_0$ that assumes reliable channels,
the parties simulate~$\pi_0$ message by message. These messages may arrive at the other side correctly or not; however, a noiseless feedback channel allows each party to learn which of its messages made it through. With this knowledge, the party tries to create a ``chain'' of correct messages. Each message contains a pointer to the last message that was not corrupted by the channel. As time goes by, the chain grows and grows, and indicates the entire correct communication of that party. 
An appealing feature of this mechanism is the fact that whenever a transmission arrives correctly at the other side, the receiver learns \emph{all} the correct transmissions so far. On the other hand, the receiver never knows whether a single received transmission (and the chain linked to it) is indeed correct.

The adversarial noise may corrupt up to $\left(\tfrac15-\eps\right)\!n$ of the messages sent by each party. We think of the adversary as one trying to construct a different, corrupt, chain. Due to its limited budget,  at the end of the coding scheme one of two things may happen. 
Either the correct chain is the longest, 
or the longest chain contains in its prefix a sufficient amount of uncorrupted transmissions. 

Indeed, if the adversary tries to create its own chain, its length is bounded by $\left(\tfrac15-\eps\right)\!n$, while the correct chain is of length $\frac{2n}{5}$ at the least.\footnote{The order of speaking in the coding scheme depends on the noise. Therefore, it is not necessary that a party speaks half of the times; see discussion below.} On the other hand, the adversary can create a longer chain which forks off the correct chain. 
As a simple example, consider the case where a party sends $\approx \frac{2n}{5}$ messages which go through uncorrupted. Now, the adversary starts corrupting the transmissions and extends the correct chain with $\left(\tfrac15-\eps\right)\!n$ corrupt messages.\footnote{This attack assumes that there are $n/5$ additional rounds where the same party speaks. This assumption is usually false and serves only for this intuitive (yet unrealistic) example.}
The corrupt forked chain is of length $\frac{2n}{5}+\left(\tfrac15-\eps\right)\!n$ and may be longer than the correct chain.
However, in this case, the information contained in the \emph{uncorrupted} prefix of the corrupt forked chain is sufficient to simulate the entire transcript of~$\pi_0$.

Another essential part of our coding scheme is its ability to alter the order of speaking according to the observed noise.\footnote{Protocols that change their length or order of speaking as a function of the observed noise are called \emph{adaptive}~\cite{GHS14,AGS16}. Since these decisions are noise-dependent, the parties may disagree on the identity of the speaker in each round, e.g., both parties may decide to speak in a given round, etc. We emphasize that due to the noiseless feedback there is always a consensus regarding whose turn it is to speak next. Hence, while our scheme has a non-predetermined  order of speaking,  the scheme is \emph{non-adaptive} by the terminology of~\cite{EGH16}; see discussion in~\cite{EGH16} and in Section~6 of~\cite{gelles17}.}
Most previous work follows the following intuition. 
If a party's transmissions were corrupted, then the information contained in these transmissions still needs to reach the other side. Therefore, the coding scheme should allow that party to speak more times.
In this work we take the opposite approach---the more a party is corrupted in the first part of the protocol, the \emph{less} it speaks in the later part. The intuition here is that if the adversary has already wasted its budget on some party, it cannot corrupt much of the subsequent transmissions of that party.
A similar approach appears in~\cite{AGS16}.

One hurdle we face in constructing our coding scheme derives from the need to communicate pointers to previous messages using a small (constant-size) alphabet. Towards this end, we first show a coding scheme that works with a large alphabet that is capable of pointing back to any previous transmission.
Next, we employ a variable-length coed, replacing each pointer with a large number of messages over a constant-size alphabet. We prove that this coding does not harm the resilience, leading to a  coding scheme with a constant-size alphabet and optimal resilience to $\left(\tfrac{1}{5}-\eps,\tfrac{1}{5}-\eps\right)$-corruptions.

\paragraph{Converse: Impossibility Bound.}
The converse proof consists of two parts. First, we show that for certain functions, any protocol  resilient to $\left(\tfrac{1}{5},\tfrac{1}{5}\right)$-corruptions must have an exponential blowup in the communication. In the second part, we show a (noisy) KW-transformation from formulas to protocols. Together, we obtain an upper bound on the noise of formulas. Indeed, assuming that there is a ``shallow'' formula that is resilient to $\left(\tfrac{1}{5},\tfrac{1}{5}\right)$-corruptions, converting it into a protocol yields a ``short'' protocol with resilience to $\left(\tfrac{1}{5},\tfrac{1}{5}\right)$-corruptions. The existence of such a protocol contradicts the bound of the first part.

The bound on the resilience of protocols 
follows a natural technique of confusing a party between two possible inputs.
We demonstrate that a $\left(\tfrac{1}{5},\tfrac{1}{5}\right)$-corruption suffices in making one party (say, Alice)
observe exactly the same transcript whether Bob holds $y$ or $y'$. Choosing $x,y,y'$
such that the output of the protocol differs between $(x,y)$ and $(x,y')$ 
leads to Alice erring on at least one of the two instances.

This idea \emph{does not} work if the protocol is allowed to communicate a lot of information. To illustrate this point, assume $f:\Sigma^n \times \Sigma^n \to \Sigma^z$ defined over a channel with alphabet~$\Sigma$. Consider a protocol where the parties send their inputs to the other side encoded via a standard Shannon error-correcting code of length~$n'=O(n)$ symbols, with distance~$1-\eps$ for some small constant~$\eps>0$. The protocol communicates $2n'$ symbols overall, and a valid $\left(\tfrac{1}{5},\tfrac{1}{5}\right)$-corruption may corrupt up to $\frac{2n'}{5}$ symbols \emph{of each one of the codewords}. However, this does not suffice to invalidate the decoding of either of the codewords, since an error-correcting code with distance~$\approx\!1$ is capable of correcting up to $\approx\!\frac{n'}{2}$ corrupted symbols. 

On the other hand, once we limit the communication of the protocol, even moderately, to around~$n$ symbols, the above encoding is not applicable anymore. 
Quite informally, our lower bound follows the intuition described below.
We show the existence of a function~$f$ such that for any protocol that computes~$f$  in $r$ rounds (where $r$ is restricted as mentioned above), the following properties hold for one of the parties (stated below, without loss of generality, for Alice).
There are inputs~$x,x',y,y'$ 
such that 
(1) $f(x,y) \ne f(x',y)\ne f(x',y')$ and 
(2) Alice speaks at most $\frac{r}{5}$ times during the first $\frac{2r}{5}$ rounds. 
Further, 
(3) when Alice holds~$x$, the protocol communicates exactly the same messages 
during its first $\frac{2r}{5}$ rounds, whether Bob holds $y$ or $y'$ (assuming no channel noise is present).

When we bound the protocol to these conditions, a $\left(\tfrac{1}{5},\tfrac{1}{5}\right)$-corruption is strong enough to make the transcript identical from Alice's point of view on $(x',y)$ and $(x',y')$, implying the protocol cannot be resilient to such an attack.
In more detail, we now describe an attack and assume Bob speaks at most $\frac{2r}{5}$ times beyond round number $\frac{2r}{5}$, given the attack. 
(If Bob speaks more, then an equivalent attack will be able to confuse Bob rather than Alice.) 
The attack changes the first $\frac{2r}{5}$ rounds as if Alice holds $x$ rather than $x'$; this amounts to corrupting at most $\frac{r}{5}$ transmissions by Alice due to property~(2). Bob behaves the same regardless of his input due to property~(3). From round~$\frac{2r}{5}$ and beyond, the attack corrupts Bob's messages so that the next $\frac{r}{5}$ symbols Bob sends are consistent with~$y$ and the following~$\frac{r}{5}$ symbols Bob communicates are consistent with~$y'$. Since Bob speaks less than $\frac{2r}{5}$ times (given the above noise), the attack corrupts at most $\frac{r}{5}$ of Bob's transmissions after round~$\frac{2r}{5}$.   

Unfortunately, while the above shows that some functions $f$ cannot be computed in a resilient manner, this argument cannot be applied towards a lower bound on resilient formulas. 
The reason is that the $KW_f$ task is not \emph{a function}, but rather \emph{a relation}---multiple outputs may be valid for a single input.
The attack on protocols described earlier shows that a $\left(\tfrac{1}{5},\tfrac{1}{5}\right)$-corruption drives the protocol to produce a different output from that in the noiseless instance. However, it is possible that a resilient protocol gives a different \emph{but correct} output.

Therefore, we need to extend the above argument so it applies to computations of arbitrary relations. Specifically, we consider the parity function on $n$ bits and its related KW-game. We show the existence of inputs that satisfy conditions (2) and (3) above, while requiring that the outputs of different inputs be disjoint; i.e., any possible output of $(x',y)$ is invalid for $(x,y)$ and for~$(x',y')$.

\medskip

The last part of the converse proof requires developing a KW-transformation from formulas to protocols, in a \emph{noise-resilience preserving} manner.
Let us begin with some background on the (standard) KW-transformation (see Section~\ref{sec:circuits:prelim} for a formal description).
The KW-game  (or rather a slight adaptation we need for our purposes) is as follows. For a boolean function $f$ on $\{0,1\}^n$, Alice gets an input $x$ such that $f(x)=0$ and Bob gets an input $y$ such that $f(y)=1$; their goal is to output a literal function $\ell(z)$ (i.e., one of the $2n$ functions 
of the form $\ell(z)=z_i$ or $\ell(z)=\neg z_i$) such that $\ell(x)=0$ and $\ell(y)=1$. 

Let $F$ be a boolean formula for $f$, consisting of $\vee$ and $\wedge$ gates, and where all the negations are pushed to the input layer (i.e., $F$ is a monotone formula of the literals $z_i$, $\neg z_i$). %
The conversion of $F$  
to a protocol $\pi$ %
for the $\KW_f$ game 
is as follows. View the formula as the protocol tree, with the literals at the bottom of the tree being the output literal function. Assign each $\wedge$-node to Alice, and each $\vee$-node to Bob. 

The invariant maintained throughout the execution of the protocol is that if the protocol reaches a node $v$, then the value of $v$ in $F$ is $0$ when evaluated on $x$, and $1$ when evaluated on~$y$. 
This invariant holds for the output gate of the formula, which is where the communication protocol begins.  
Next, each time that the protocol is at node~$v$ and it is Alice's turn to speak (thus $v$ is an $\wedge$-gate in $F$), Alice sends the identity of a child which evaluates to $0$ on $x$. Note that assuming the invariant holds for $v$, Alice can send the identity of such a child (since at least one of the inputs to an AND gate that outputs a $0$, also evaluates to~$0$), while this child must evaluate to~$1$ on~$y$ assuming $v$ evaluates to~$1$ on~$y$. 
By maintaining this invariant, Alice and Bob arrive at the bottom, where they reach a literal evaluating to~$0$ on~$x$ and $1$ on~$y$. 
Note that there is some room for arbitrary decision making: if more than one child of $v$ evaluates to $0$ on~$x$, Alice is free to choose any such child---the protocol will be valid for any such choice.

In this work we extend the above standard KW-transformation to the noisy-regime. 
Namely, we wish to convert a \emph{resilient} formula into an interactive protocol~$\pi$ 
\emph{while keeping the protocol resilient} to a similar level of channel noise.
We note that the extension we need is completely different from what is found in previous uses of the KW-transformation. 
Indeed, for the achievability bound, a KW-transformation is used in both steps (i) and (iii) in the above outline of~\cite{KLR12}.
However,
the instance used in step (i) assumes there is no noise, while the instance in step (iii) works in the other direction, i.e., it transforms (resilient) protocols to (resilient) formulas.

Similar to the standard transformation, our noisy KW-transformation  starts by 
constructing a protocol tree based on the formula's structure, where  every $\AND$-gate is assigned to Alice and every $\OR$-gate to Bob. The main difference is in the decision making of how to 
proceed when reaching a node~$v$. 
The goal is to keep the invariant 
that the gate~$v$ in $F$ evaluates to~$0$ on~$x$ and to~$1$ on~$y$, 
\emph{even when noise is present}.

When only one of $v$'s descendants  evaluates to~$0$ on~$x$ in~$F$, 
Alice has no choice but to choose that child. 
However, when more than a single descendant evaluates to~$0$ on~$x$, 
Alice's decision is less obvious. 
Moreover, this decision may affect the resilience of the protocol---it is possible that noise causes one of the descendants evaluate to~$1$ on that given~$x$. 

We observe, however, that one of $v$'s children evaluates to 0 on~$x$ given \emph{all the noise patterns} $F$ is resilient against. The other children may still evaluate to 1 sometimes, as a function of the specific noise. Once we identify this special child that always evaluates to~$0$, Alice can safely choose it and maintain the invariant (and the correctness of the protocol), regardless of future noise.
In more detail, we prove that if such a special child did not exist and all descendants  could evaluate to both~0 and~1 as a function of the noise, 
then we could construct a noise pattern~$E^*$ that would make all descendants evaluate to 1 on~$x$ simultaneously. 
Hence, assuming the noise is $E^*$, the node $v$ would evaluate to~1 on $x$, and consequently $F(x)=1$. 
At the same time, we show that $F$ is resilient to the noise~$E^*$, 
so $F(x)=0$ assuming the noise is~$E^*$, and we reach a contradiction.

\subsection{Other related work}
The field of interactive coding schemes~\cite{gelles17} started with the seminal line of work by Schulman~\cite{schulman92,RS94,schulman96}. Commonly, the goal is to compile interactive protocols into a noise-resilient version that has (1) good noise resilience; (2) a good rate; and (3) high probability of success. Computational efficiency is another desired goal.
Numerous works achieve these goals, either fully or partially~\cite{BR14,GMS14,BKN14,FGOS15,BE17,GH14,KR13,haeupler14,GHKRW18}, where the exact parameters depend on the communication and noise model. 

Most related to this work are coding schemes in the setting where a noiseless feedback channel is present.
Pankratov~\cite{pankratov13} gave the first interactive coding scheme that assumes noiseless feedback. The scheme of~\cite{pankratov13} aims to maximize its rate,  assuming all communication passes over a binary symmetric channel (BSC) with flipping parameter~$\eps$ (i.e., a channel that communicates bits, where every bit is flipped with probability~$\eps$, independently of other bits). Pankratov's scheme achieves a rate of $1-O(\sqrt\eps)$ when $\eps\to0$. Gelles and Haeupler~\cite{GH17} improved the rate in that setting to~$1-O(\eps\log1/\eps)$, which is the current state of the art.
For the regime of high noise, Efremenko, Gelles, and Haeupler~\cite{EGH16} provided coding schemes with maximal noise resilience, assuming noiseless feedback. They showed that the maximal resilience depends on the channel's alphabet size and on whether or not the order of speaking is noise-dependent. Specifically, they developed coding schemes with a noise-independent order of speaking and a constant rate that are resilient to $1/4-\eps$ and $1/6-\eps$ fractions of noise with a ternary and binary alphabet, respectively. When the order of speaking may depend on the noise, the resilience increases to $1/3-\eps$ for any alphabet size. They showed that these noise levels are optimal and that no general coding scheme can resist higher levels of noise.

There has been a tremendous amount of work on coding  for noisy channels with noiseless feedback in the one-way (non-interactive) communication setting, starting with the works of Shannon, Horstein, and Berlekamp~\cite{shannon56,Horstein63,berlekamp1964block}. It is known that the presence of feedback does not change the channel's capacity, however, it improves the error exponent. The maximal noise-resilience in this setting is also known. Recently, Haeupler, Kamath, and Velingker~\cite{HKV15} considered deterministic and randomized codes that assume a partial presence of feedback.

\subsection{Organization}
The first half of our paper considers interactive coding protocols over noisy channels with noiseless feedback.
Section~\ref{sec:imp} proves that any interactive coding scheme that is resilient to $\left(\frac15,\frac15\right)$-corruptions must exhibit a zero rate.
Sections \ref{sec:codingLarge}--\ref{sec:codingSmall} 
describe our constant-rate coding scheme that is resilient to $\left(\frac15-\eps,\frac15-\eps\right)$-corruptions. 
First, Section~\ref{sec:codingLarge} describes a scheme with a large alphabet (polynomial in the length of the protocol).
Then, Section~\ref{sec:codingSmall} shows how to reduce the alphabet to a constant size.

The second half of the paper (Section~\ref{sec:circuits}) considers noise-resilient circuits.
First, in Section~\ref{sec:circuits:prelim} we recall  
the notions of formulas, short-circuit noise and the (noiseless) KW-transformation.
In Section~\ref{sec:FtoP} 
we present our noise-preserving KW-transformation and show how to convert a resilient formula into a resilient protocol.
This reduction (along with the impossibility from Section~\ref{sec:imp}) proves the converse theorem, showing that the resilience we obtain for formulas is maximal. 
In Section~\ref{sec:PtoF} we provide the other direction, a noise-resilient transformation from protocols to formulas (following~\cite{KLR12}).
Employing the coding scheme of Section~\ref{sec:codingSmall} we give an efficient method that compiles any formula into an optimal resilient version. 

\section{Preliminaries}\label{sec:prelim}
\paragraph{Notations}
For integers $i\le j$ we denote by $[i,j]$ the set $\{i,i+1,\ldots, j\}$ and by $[i]$ the set $\{1,\ldots, i\}$.
We let~$\Sigma$ be some finite set. For a string $s\in \Sigma^*$ and two indices $x,y \in \{1,\ldots, |s|\}$, $x<y$ we let $s[x,y]=s_{x}s_{x+1}\cdots s_{y}$. We will treat $\emptyset$ as the empty word, i.e., for any $a\in \Sigma^*$ we have $a\circ \emptyset = \emptyset\circ a = a$, where $\circ$ stands for concatenation. 
For bits $a,b\in\{0,1\}$ we denote  $a\oplus b=a+b \mod 2$, $a \wedge b=a \cdot b$, and $\overline b = 1-b$.
For two bitstrings of the same length $x,y\in\{0,1\}^n$ we denote by 
$\langle x,y\rangle=\bigoplus_i (x_i\cdot y_i)$ their inner product (mod~2) as vectors over~$GF(2)$.
We denote $\overline x= \overline x_1\overline x_2\cdots \overline x_n$, the bit-wise complement of~$x$.
All logarithms are taken to base 2, unless the base is explicitly written.

\paragraph{Interactive Protocols}
In the interactive setting we have two parties, Alice and Bob, who receive private inputs $x\in X$ and $y\in Y$, respectively. Their goal is to compute some predefined function $f(x,y): X \times Y \to Z$ by sending messages to each other.
A \emph{protocol} describes for each party the next message to send, given its input and the communication received so far. We assume the parties send symbols from a fixed alphabet $\Sigma$. %
 The protocol also determines when the communication ends and the output value (as a function of the input and received communication).

Formally, an interactive protocol $\pi$ can be seen as a $|\Sigma|$-ary tree (also referred to as the \emph{protocol tree}), where each node $v$ is assigned either to Alice or to Bob. For any $v$ node assigned to Alice there exists a mapping $a_v: X \to \Sigma$ that maps the next symbol Alice should send, given her input. Similarly, for each one of Bob's nodes we set a mapping $b_v: Y \to \Sigma$. Each leaf is labeled with an element of~$Z$.
The output of the protocol on input $(x,y)$ is the element at the leaf reached by starting at the root node, and traversing down the tree, where, at each internal node $v$ owned by Alice (resp., Bob), if $a_v(x)=i$ (resp., $b_v(y)=i$) the protocol advances to the $i$-th child of $v$.
For convenience, we denote  Alice's nodes by the set~$V_a$ and Bob's nodes by the set~$V_b$.
We may assume that all the nodes in a given protocol tree are reachable by some input $(x,y)\in X\times Y$ (otherwise, we can prune that branch without affecting the behavior of the protocol). 
Note that the order of speaking in $\pi$ does not necessarily alternate and it is possible that the same party is the sender in consecutive rounds. For any given transcript~$T$, we denote by~$\pi( \cdot \mid T)$ the instance of~$\pi$ assuming the history~$T$. Specifically, assuming Alice is the sender in the next round (assuming the history so far is~$T$), then the next communicated symbol is~$\pi(x \mid T)$.

The length of a protocol, denoted~$|\pi|$, is the length of the longest root-to-leaf path in the protocol tree, or equivalently, it is the maximal number of symbols the protocol communicates in any possible instantiation. In the following we assume that all instances have the same length~$|\pi|$.
The communication complexity of the protocol is 
\[
CC(\pi)= |\pi| \log |\Sigma| .
\]
When $\Sigma$ is constant (independent of the input size), we have $CC(\pi) = O(|\pi|)$.
If, by round $t$, Alice is the sender in~$t_A$ rounds and Bob is the sender in~$t_B=t-t_A$ rounds, 
we denote their respective communication complexity until round~$t$ by
$CC_A^{\le t}(\pi)=t_A \log |\Sigma|$ and $CC_B^{\le t}(\pi)=t_B \log |\Sigma|$.

\paragraph{Transmission Noise with Feedback}
We will assume the communication channel may be noisy, that is, the received symbol may mismatch with the sent symbol. All the protocols considered in this work assume the setting of \emph{noiseless feedback}: the sender always learns the symbol that the other side received (whether corrupted or not). The receiver, however, does not know whether the symbol it received is indeed the one sent to him.

A noise pattern is defined as $E \in \{0,1,\dotsc,|\Sigma|-1,\noerr\}^{|V_a| \cup |V_b|}$. 
For any node $v$, $E_v$ denotes the symbol that the receiver gets for the transmission that is done when the protocol reaches the node~$v$. Specifically, say $v$ is an Alice-owned node, then if $E_v=\noerr$, Bob receives the symbol sent by Alice; otherwise, $E_v\ne\noerr$, Bob receives the symbol $E_v$. Note that due to the feedback, Alice learns that her transmission was corrupted as well as the symbol that Bob received, and the protocol descends to the node dictated by $E_v$. We denote by $\pi_E$ the protocol $\pi$ when the noise is dictated by~$E$; we sometimes write $\pi_0$ for a run of the protocol with no transmission noise, i.e., with the pattern $E=\noerr^{|V_a| \cup |V_b|}$.

We say that a protocol is \emph{resilient} to a noise pattern~$E$ if for any $(x,y)\in X\times Y$ it holds that $\pi_E$ outputs the same value as $\pi_0$. 
While it is common to limit the noise to a constant fraction of the transmissions, 
in this work we take a more careful look at the noise, and consider the exact way it affects the transmissions of each party.
\begin{definition}
An $(\alpha,\beta)$-corruption is a noise pattern that changes at most $\alpha |\pi|$ symbols sent by Alice and at most $\beta |\pi|$ symbols sent by Bob. Note that the effective (combined) noise rate is $(\alpha+\beta)$.
\end{definition}

\section{Resilience to $(1/5,1/5)$-Corruptions is Impossible}
\label{sec:imp}

In this section we prove that no coding scheme with constant overhead can be resilient to a $(1/5,1/5)$-corruption.
To this end we show a specific $(1/5,1/5)$-corruption that confuses any protocol for a specific function~$f$ that is
``hard'' to compute in linear communication. Our result \emph{does not} apply to coding schemes with vanishing rates. In fact, if the communication is exponentially large, coding schemes with resilience higher than $1/5$ exist.\footnote{For instance, consider the scheme in which each party sends its input to the other side encoded via a standard (Shannon) error-correcting code with distance $\approx 1$. This trivial protocol is resilient to $(1/4-\eps,1/4-\eps)$-corruption, yet its rate is~0.}

Normally,  we discuss the case where protocols compute a \emph{function} $f: X\times Y \to Z$.
While our converse bound on the resilience of interactive protocols works for some hard function (e.g., the pointer jumping), such a proof does not suffice towards our converse on the resilience of boolean formulas (Theorem~\ref{thm:converse-inf}). 
The reason is that the conversion from formulas to protocols does not yield a protocol that computes a function, but rather a protocol that computes a \emph{relation}.
Recall that for any given function~$f$ and any input $(x,y)$ such that $f(x)=0$ and $f(y)=1$,
the KW-game for~$f$, $KW_f$, outputs an index $i\in [n]$ for which $x_i \ne y_i$ (see Section~\ref{sec:circuits:prelim} for a formal definition). 
However, multiple such indices may exist and each such an index is a valid output.

Let $X,Y,Z$ be finite sets and $R \subseteq X\times Y\times Z$ be a ternary relation.
For any $(x,y)\in X\times Y$ and a given relation~$R$ 
let $R(x,y)=\{z \mid (x,y,z)\in R\}$ be the set of all $z$ that satisfy the relation for $x,y$.
We assume that for any $x,y$ it holds that $|R(x,y)|>0$.
Given such a relation, a protocol that \emph{computes the relation} is the following two-party task.
Alice is given $x\in X$ and Bob is given~$y\in Y$. The parties need to agree on some $z\in R(x,y)$. 

\smallskip

We now show an explicit relation for which no protocol (of ``short''' length) is resilient to $(1/5,1/5)$-corruptions. Specifically, in the rest of this section we consider the binary parity function on $n$ bits, $par:\{0,1\}^n \to \{0,1\}$, defined for any $x\in\{0,1\}^n$ by 
\[
par(x)=x_1\oplus \dotsm \oplus x_n.
\]
Let $X= \{ x\in\{0,1\}^n \mid par(x)=0\}$ and $Y=\{ y\in\{0,1\}^n \mid par(y)=1\}$. 
We let $KW_{par} \subseteq X \times Y \times[n]$ 
be the KW-game for the parity function, 
defined by
\[
KW_{par} = \left\{ (x,y,z)\mid  par(x)=0 \wedge par(y)=1 \wedge x_z \ne y_z \right \}.
\]

We will need the following technical claim.

\begin{claim}\label{clm:inner-prod}
Let $Y\subseteq \{0,1\}^n$.
If $|Y| \geq 2^{n/2}+1+n/2$ then there exist two distinct elements $y_1, y_2\in Y$ such that $\langle y_1, y_2\rangle =1$. Furthermore, if 
$|Y|\geq 2^{(n+1)/2}+2+n/2$ then there  exist two distinct elements $y_1, y_2\in Y$ such that $\langle y_1, y_2\rangle =0$.
\end{claim}
\begin{proof} 
Since $|Y| > 2^{n/2}$, there exist $k=\lfloor n/2\rfloor +1$ linearly independent elements  $b_1, b_2, \ldots, b_{k} \in Y$.  Consider the linear space $L=\text{span}\{b_i\}_{i=1}^{k}$. 
Let $L^{\perp}=\{ v\in \{0,1\}^n \mid \langle v,w\rangle=0\text{ for all }w\in L\}$ be the orthogonal space of~$L$ with respect to the $\langle \cdot ,\cdot \rangle$ product and recall that
$\dim L + \dim L^{\perp} = n$.
Since $\dim L=k$ we get that $\dim L^{\perp} = n-k < n/2 $ and therefore $|L^{\perp}|<2^{n/2}$. Consider $Y'=Y\setminus \{b_i \mid  i\in [k] \}$. Since $|Y'|\geq 2^{n/2}$ there must exist $y\in Y'$ such that $y\notin L^{\perp}$, which, in turn means that its product with at least one of the $b_i$s must be non-zero (or otherwise $y$ would belong to the orthogonal space~$L^{\perp}$). That is, there exists $b_i$ for which $\langle b_i, y\rangle =1$, as stated.

For the second part of the claim, let us construct $\tilde{Y}= \{(y,1)\in   \{0,1\}^{n+1} \mid  y\in Y\}$; this is merely the set $Y$ with an additional coordinate which is always set to one (over a space of dimension~$n+1$). 
Note that $|\tilde Y|= |Y| \geq 2^{(n+1)/2}+1+ (n+1)/2$ and we can use the first part of this claim to show that  
there exist two distinct elements  $ (y_1,1), (y_2,1) \in \tilde{Y}$ such that $\langle (y_1,1), (y_2,1)\rangle =1$; therefore, $\langle y_1, y_2\rangle =0$.   
\end{proof}

\begin{lemma}\label{lem:KwHard}
Let $\pi$ be an interactive protocol for $KW_{par}$ (with inputs of $n$ bits)
of length $|\pi| = r$ defined over a communication channel with alphabet~$\Sigma$ and noiseless feedback. Without loss of generality, let Alice be the party who speaks less in the first~$2r/5$ rounds of~$\pi$ (averaging over all possible inputs $(x,y)\in X\times Y$).
 Additionally, assume $n/3 > 2r\log(2|\Sigma|)/5$.

Then, there exist distinct inputs $x,x'\in X$, $y,y'\in Y$ for which:
\begin{description}
\item[(1)] $\pi(x,y)$ and $\pi(x,y')$ agree on the first $2r/5$ rounds.
\item[(2)] During the first $2r/5$ rounds of the execution $\pi(x,y)$ Alice speaks fewer times than Bob. 
\item[(3)] $KW_{par}(x',y) \cap KW_{par}(x',y') = \emptyset$ and 
$KW_{par}(x',y) \cap KW_{par}(x,y) = \emptyset$.
\end{description}
\end{lemma}
Note that the above lemma assumes Alice is the party that speaks fewer times in the first $2r/5$ rounds of~$\pi$
when averaging over all possible inputs $(x,y)\in X\times Y$; otherwise, a symmetric lemma holds for Bob.
\begin{proof}
Let $x$ be an input for Alice such that on most of the values~$y$, Alice speaks fewer times in the first $2n/5$ rounds of~$\pi(x,y)$. Such an input must exist by our choice of Alice.
Let 
\[
Y'=\left\{ y\in Y {\,}\middle\vert{\,} CC_A^{\le 2r/5}(\pi(x,y))\le CC_B^{\le2r/5}(\pi(x,y))\right\}
\] 
be the set of all inputs for Bob, where Alice speaks fewer times in the first $2r/5$ rounds of~$\pi$ assuming Alice holds the above~$x$. By the choice of~$x$, it holds that $|Y'| \ge 2^n/2$.

Consider the set of transcript prefixes of length~$2r/5$ generated by~$\pi$ when Alice holds the above~$x$ and Bob holds some input from the set~$Y'$, 
\[
T_x=\left\{t[1,2r/5] \mid t=\pi(x,y),  y \in Y'\right\}.
\] 
Note that there are at most $(2|\Sigma|)^{2r/5}$ different prefixes of length $2r/5$ over $\Sigma$ with an arbitrary order of speaking.
Since we assumed $n/3 > 2r\log(2|\Sigma|)/5$, we have, for large enough~$n$,
\[
|Y'|\ge 2^{n-1} \ge (2^{(n+1)/2}+1) 2^{n/3} \ge (2^{(n+1)/2}+1) 2^{2r\log(2|\Sigma|)/5} \ge  \Upsilon |T_x|,
\] 
with $\Upsilon = 2^{(n+1)/2}+1$.
Using a pigeon-hole principle, there must be $y^1,y^2,\ldots , y^{\Upsilon}\in Y'$ such that 
$\{\pi(x,y^i)\}_{i=1}^{\Upsilon}$ agree on the first 
 $2r/5$ rounds of the protocol---they have an identical order of speaking and they communicate the same information. 
 
 Next consider the set $\{\overline{x}\oplus y^i\}_{i=1}^{\Upsilon}$.
Claim~\ref{clm:inner-prod} guarantees that there exist two elements in that set such that 
\[
\langle \overline{x}\oplus y^i, \overline{x}\oplus y^j\rangle =par( \overline{x});
\]
these $y^i, y^j$ will be our $y, y'$.

Note that Properties (1) and (2) of the lemma are satisfied by the above $x,y,y'$. We are left to show an input $x'$ for Alice that satisfies property~(3).

Based on the above $x,y,y'$ we construct $x'$ in the following manner. For any $i\in [n]$ set
\[
x'_i = 
\begin{cases}
y_i & y_i = y'_i \\
\overline{x_i} & y_i \ne y'_i
\end{cases}.
\]
The above $x'$ is constructed such that outputs given by $KW_{par}$ are disjoint if we change only the input of Alice or only the input of Bob. Formally,
\begin{claim}The following claims hold for the above $x,x',y,y'$,
\begin{enumerate}[label=(\alph*)]
\item $par(x') =0$,
\item $KW_{par}(x',y) \cap KW_{par}(x',y') = \emptyset$,
\item  $KW_{par}(x',y) \cap KW_{par}(x,y) = \emptyset$  and $KW_{par}(x',y') \cap KW_{par}(x,y') = \emptyset$.
\end{enumerate}
\end{claim}
\begin{proof}
\begin{enumerate}[label=(\alph*)]
\item It is easy to check that 
$x'_i=  \left((\overline{x_i}\oplus y_i) \cdot (\overline{x_i}\oplus y'_i)\right)\oplus \overline{x_i} $. Therefore, 
\begin{align*}
par(x') 
&=\bigoplus_{i=1}^n x'_i  \\
&=\bigoplus_{i=1}^n\left ( \left((\overline{x_i}\oplus y_i) \cdot (\overline{x_i}\oplus y'_i)\right)\oplus \overline{x_i}\right ) \\
&=\langle \overline{x}\oplus y, \overline{x}\oplus y'\rangle \oplus par(\overline{x}). 
\end{align*}
Since we picked $y, y'$ for which $\langle \overline{x}\oplus y, \overline{x}\oplus y'\rangle = par(\overline{x})$, we conclude that $par(x')=0$.
\item Assume towards contradiction that $i\in KW_{par}(x',y) \cap KW_{par}(x',y')$, i.e., $x'_i \ne y_i$ as well as $x'_i \ne y'_i$. However, $x'_i,y_i,y'_i$ are all bits and these two inequalities imply $y_i = y'_i$. But then, $x'_i = y_i$ by the way we construct~$x'$, which is a contradiction.
\item  
Assume towards contradiction that $i \in KW_{par}(x',y) \cap KW_{par}(x,y)$. That is, 
$x'_i \ne y_i$ and $x_i \ne y_i$, which means that $x'_i = x_i$. 
On the other hand, by the construction of $x'$, 
either $x'_i \ne x_i$ or $x'_i = y_i$. Both options lead to a contradiction.
The proof of the second part is identical.
\end{enumerate}
\end{proof}
The first claim proves that $x'\in X$. The other claims prove property (3) of the lemma and conclude its proof.
\end{proof}

Our main result in this section is the following theorem, proving that no protocol for the $KW_{par}$  can be resilient to a $(1/5,1/5)$-corruption if its communication is bounded. This will imply that any coding scheme that is resilient to $(1/5,1/5)$-corruption must have rate~0.
Specifically, it cannot produce a protocol with a constant overhead with respect to the optimal protocol that computes $KW_{par}$ over reliable channels.
\begin{theorem}\label{thm:KWnotresilient}
Any interactive protocol $\pi$ that computes the relation $KW_{par}$
with at most $|\pi|<\frac56\frac{n}{\log(2|\Sigma|)}$ rounds over a noisy channel with alphabet~$\Sigma$ and noiseless feedback is not resilient to $(1/5,1/5)$-corruptions.
\end{theorem}
\begin{proof}
Let $\pi$ be a protocol 
with $r<\frac56\frac{n}{\log(2|\Sigma|)}$ rounds communicating symbols from the alphabet~$\Sigma$.
Via Lemma~\ref{lem:KwHard}, let $x_0,x_1\in X$ and $y_0,y_1\in Y$ 
be inputs that satisfy:
\begin{description}
\item[(1)] $\pi(x_0,y_0)$ and $\pi(x_0,y_1)$ agree on the first $2r/5$ rounds.
\item[(2)] During the first $2r/5$ bits of the protocol $\pi(x_0,y_0)$ Alice speaks less than Bob. 
\item[(3)] $KW_{par}(x_1,y_0)\cap KW_{par}(x_1,y_1)=\emptyset$ and $KW_{par}(x_1,y_0)\cap KW_{par}(x_0,y_0)=\emptyset$.
\end{description}

\goodbreak
We now generate a transcript~$T$ and show that
$T$ is consistent with a $(1/5,1/5)$-corruption of $\pi(x_1,y_0)$.
Additionally, it is either the case that  $T$ is consistent with   
 a $(1/5,1/5)$-corruption of $\pi(x_1,y_1)$ or it is  consistent with a
$(1/5,1/5)$-corruption of $\pi(x_0,y_0)$. In the first case, Alice is unable to distinguish the case where Bob holds~$y_0$ and $y_1$; in the second, Bob cannot tell if Alice holds $x_0$ or $x_1$. The outputs for different inputs are distinct by property~(3). Thus the confused party is bound to err on at least one of them.

Note that the transcript~$T$ contains messages \emph{received} by the two parties, which may be noisy. Due to the feedback, both parties learn~$T$. Additionally, the order of speaking in~$\pi$ is entirely determined by (prefixes of)~$T$. Specifically, if two different instances of $\pi$ have the same received transcript by round~$j$, the party to speak in round~$j+1$ is identical in both instances.

\goodbreak
\smallskip
\noindent The string~$T$ is obtained in the following manner:
{\small
\begin{enumerate}
\item Run $\pi(x_0,y_0)$ for $2r/5$ rounds. Let $T_1$ be the generated transcript.
\item Run $\pi(x_1,y_0 \mid T_1)$ until Bob transmits $r/5$ additional symbols (unless $\pi$ terminates beforehand). Let $T_2$ be the generated transcript.  
\item (if $|T_1T_2|<r$) Run $\pi(x_1,y_1 \mid T_1T_2)$ until Bob transmits $r/5$ additional symbols (unless $\pi$ terminates beforehand). 
\item (if $|T_1T_2T_3|<r$) Let $T_4$ describe $\pi(x_1,y_0 \mid T_1T_2T_3)$ until it terminates. 
\item Set $T=T_1T_2T_3T_4$.  
\end{enumerate}
In the case where the above algorithm did not execute Step~$i$, for $i\in\{3,4\}$, assume $T_i=\emptyset$.
} %

We now show that $T$ corresponds to a $(1/5,1/5)$-corrupted execution of $\pi$ for two different valid inputs with disjoint outputs.  
We consider two cases: (i) when Step 3 halts since $T$ reached its maximal size of $r$ symbols  (i.e., when $T_4=\emptyset$), and (ii) when Step 3 halts since Bob transmitted $r/5$ symbols in this step ($T_4\ne\emptyset$).
\begin{description}
\item[case (i)  $\mathbf{T_4=\emptyset}$.]
In this case we show that a $(1/5,1/5)$-corruption suffices to make the executions of $\pi(x_1,y_0)$ and $\pi(x_1,y_1)$ look the same from Alice's point of view.

Let $\Pi$ be the transcript of a noisy execution of $\pi(x_1,y_0)$ (defined shortly) and split $\Pi$ into three parts: $\Pi=\Pi_1\Pi_2\Pi_3$ that correspond in length to $T_1,T_2,T_3$.
The noise changes all Alice transmissions in $\Pi_1$ so that they correspond to Alice's symbols in~$T_1$; the noise changes all Bob's transmissions in $\Pi_3$ so that they correspond to Bob's transmissions in~$T_3$.
It is easy to verify that the obtained transcript $\Pi$ of received messages is exactly~$T$.
Furthermore, the first part changes at most $r/5$ transmissions by Alice, since by property~(2) Alice speaks fewer times in the first $2r/5$ of the instance $\pi(x_0,y_0)$.  
The second part changes at most $r/5$ transmissions of Bob since $T_3$ halts before Bob communicates additional $r/5$~transmissions. Hence, the noise described above is a valid $(1/5,1/5)$-corruption.

On the other hand, and abusing notation, consider a (noisy) instance of~$\pi(x_1,y_1)$ and let $\Pi=\Pi_1\Pi_2\Pi_3$ be the received messages transcript split into parts that corresponds in length to $T_1,T_2,T_3$, assuming the following noise.
Again, the noise changes all Alice's  transmissions in $\Pi_1$ to be the corresponding symbols received in~$T_1$. 
This makes the $2r/5$ first rounds of the received transcript look like the ones in the instance~$\pi(x_0,y_1)$.
By Property~(1), these transmissions agree with the first $2r/5$ transmissions in the noiseless instance~$\pi(x_0,y_0)$; hence, the corrupted $\Pi_1$ equals~$T_1$. 
Next, the noise changes Bob's transmissions in $\Pi_2$ to correspond to~$T_2$. The obtained transcript~$\Pi$ is then exactly~$T$  (note that $\Pi_3=T_3$ by definition). 
Again, $T_1$ contains at most $2r/5$ of Alice's transmissions, and $T_2$ contains at most $r/5$ transmissions of Bob by definition. Hence, this is a valid $(1/5,1/5)$-corruption.

We conclude by recalling that $KW_{par}(x_1,y_0) \cap KW_{par}(x_1,y_1)=\emptyset$, that Alice must be wrong on at least one of the above executions, since her view in both executions is the same.
Note that the above proof holds even when $T_3=\emptyset$.

\item[case (ii) $\mathbf{T_4 \ne \emptyset}$.]
In this case we show a $(1/5,1/5)$-corruption that makes the executions of $\pi(x_0,y_0)$ and $\pi(x_1,y_0)$ look the same from Bob's point of view.
We point out that Alice speaks at most $r/5$ times after Step~1. 
Indeed, Step~1 contains $2r/5$ rounds, and Steps 2--3 contain $2r/5$ rounds where Bob speaks, hence, Alice may speak in at most another $r/5$ times after Step~1.

Let $\Pi$ be the transcript of a noisy execution of $\pi(x_0,y_0)$ where the noise is defined below. Split $\Pi$ into 4 parts: $\Pi=\Pi_1\Pi_2\Pi_3\Pi_4$ that correspond in length to $T_1,T_2,T_3,T_4$.
The noise changes all Alice's transmissions in $\Pi_2\Pi_3\Pi_4$ so that they match the corresponding symbols of $T_2,T_3,T_4$. As mentioned, this corrupts at most $r/5$ symbols.
Additionally, the noise changes Bob's transmissions in $\Pi_3$ to correspond to $T_3$; 
this, by definition, entails $r/5$ corruptions of Bob's transmissions.
The obtained transcript $\Pi$ is exactly~$T$.

On the other hand, and abusing notation again, consider a noisy execution of $\pi(x_1,y_0)$ denoted by $\Pi=\Pi_1\Pi_2\Pi_3\Pi_4$. Here the noise is defined as follows. 
The noise changes all Alice's transmissions in $\Pi_1$ to match the corresponding symbols of~$T_1$. 
As before, the noise changes Bob's transmissions in $\Pi_3$ to match~$T_3$.
Now it holds that $\Pi=T$, while the noise corrupted at most $r/5$ of each party's transmissions.

We conclude by recalling that $KW_{par}(x_0,y_0) \cap KW_{par}(x_1,y_0)=\emptyset$. Thus, Bob must be wrong on at least one of the above executions, since his view in both executions is exactly the same.
\end{description}
\end{proof}

Note that $KW_{par}$ has a protocol of length $O(\log n)$ assuming reliable channels.\footnote{This can easily be seen, e.g.,  by considering a formula that computes the parity of $n$ bits, and applying the Karchmer-Wigderson transformation~\cite{KW90}.}
Theorem~\ref{thm:KWnotresilient} leads to the following conclusion.

\begin{corollary}\label{cor:imp1/5}
There exists an interactive  protocol~$\pi_0$ defined over a noiseless channel with feedback 
such that any protocol~$\pi$ that computes the same functionality as~$\pi_0$ and
is resilient to $(1/5,1/5)$-corruptions (assuming noiseless feedback)
must incur an exponential blowup in the communication.
\end{corollary}
As a consequence, any coding scheme that compiles any protocol into a $(1/5,1/5)$-resilient version  must have rate~$0$.

\section{A Coding Scheme with a Large Alphabet}
\label{sec:largeAB}\label{sec:codingLarge}
In this section we construct a coding scheme for interactive protocols assuming noiseless feedback. 
We show that for any constant $\eps>0$,
any protocol~$\pi_0$ defined over noiseless channels (with noiseless feedback) 
can be simulated by a protocol~$\pi= \pi_\eps$ defined over noisy channels (with noiseless feedback) 
such that   
(1) $CC(\pi)/CC(\pi_0) = O_\eps(1)$, and 
(2) $\pi$ is resilient to $(1/5-\eps,1/5-\eps)$-corruptions.
The protocol~$\pi$ in this section communicates symbols from a large alphabet of polynomial size in~$|\pi_0|$. In later sections we show how to reduce the size of the alphabet. 
While the coding scheme $\pi$ will alter its order of speaking in accordance with the noise, we will assume that $\pi_0$ is an alternating protocol. This is without loss of generality, since any protocol can be made alternating by increasing its communication complexity by a factor of at most~2.

\subsection{The coding scheme}

\begin{figure}[htp]
\centering
\begin{tikzpicture}
[msg/.style={rectangle,draw=black,thick,minimum height=0.33cm,minimum width=0.66cm},
cor/.style={rectangle,draw=red,thick,text=red,minimum height=0.33cm,minimum width=0.66cm},
wrong/.style={rectangle,draw=black,fill=red!50,thick,minimum height=0.33cm,minimum width=0.66cm}
]

\node at (1,3) (m1) [msg] {$m_1$};
\node at (3,3) (m3) [cor] {$\mathbf{X}$};
\node at (5,3) (m5) [msg] {$m_5$};

\node at (2,1) (m2) [msg] {$m_2$};
\node at (4,1) (m4) [wrong] {$m_4$};
\node at (6,1) (m6) [msg] {$m_6$};

\path (m1)  edge [<-,bend left,very thick,red] (m3);
\path (m1)  edge [<-,bend left=50,thick] (m5);
\path (m2)  edge [<-,bend left,thick] (m4);
\path (m4)  edge [<-,bend left,thick] (m6);

\path node [left] {round} (0.5,0) edge [->]   (8,0);
\foreach \i in {1,2,...,7}
{
	\draw (\i,-0.1) -- (\i,0.1);
	\node at (\i, -0.5) {\i};
}

\end{tikzpicture}
\caption{Alice is the sender at odd rounds, and Bob at even rounds. 
Each message $m_i=(link_i,b_i)$ contains a link to the previous \emph{uncorrupted} message sent by the same party and the next bit of $\pi_0$ according to current information.
The message $m_3=(m_1,b_3)$ sent by Alice in round~$3$ was corrupted into $X=(m_1,\overline b_3)$. 
At the beginning of round~$4$, the current chain contains ${m_1,X}$ for Alice and $m_2$ for Bob, thus Bob sends $m_4$ according to this false information. Note that Bob's message is not corrupted, yet Alice knows it is based on wrong information and ignores it. 
In round 5 Alice indicates to Bob that round 3 was corrupted by linking $m_5$ to $m_1$ rather than to the corrupted~$X$. Bob learns that $m_4$ was generated based on wrong information and can correct it using future transmissions. In round 6, the current chain contains $m_1,m_5$ for Alice and $m_2,m_4,m_6$ for Bob. The information in $m_4$, although on Bob's chain, will not affect the generated simulated transcript.}
\end{figure}

At a high level, the coding scheme (Algorithm~\ref{alg:coding}) runs for $n=|\pi_0|/\eps$ rounds in which it tries to simulate~$\pi_0$ step by step. The availability of noiseless feedback allows a party to notice when 
the channel alters a transmission sent by that party.
The next time that party speaks, it will re-transmit its message and ``link'' the new transmission to its latest uncorrupted transmission. That is, each message carries a ``link''---a pointer to a previous message sent by the same party. By following the links, the receiver learns the ``chain'' of  uncorrupted transmissions;  the party considers all ``off-chain'' transmissions as corrupted.
Note that there are two chains: one for symbols received by Alice and one for symbols received by Bob. However, due to the feedback, \emph{both} parties learn the received symbols at both sides and can infer both chains.

The algorithm consists of several sub-procedures.
The $\Parse$ procedure (line~\ref{alg:Parse}) %
parses all the transmissions received so far at a given party and outputs the ``current chain'' of that party: the (rounds of the) transmissions linked by the \emph{latest} received transmission. 
Note that  once a new transmission arrives, the current chain of the recieving party possibly changes. Moreover, upon reception of a corrupt transmission, a corrupt chain may be retrieved.

Given the current chains of both parties, the parties can infer a candidate for a partial transcript for~$\pi_0$. We call this transcript the \emph{simulated} transcript of~$\pi_0$ given the current chains.
Note, that if the last received transmission is corrupted, the current chain linked to it might be arbitrary, and the candidate for the simulated transcript will be wrong.
The \textsc{TempTranscript} procedure (line~\ref{alg:TempTrans}) determines the partial simulated transcript of~$\pi_0$ according to all the messages received in~$\pi$ so far at both sides, i.e., according to the current chains.
Again, the scheme considers only transmissions that are on the current chains and ignore all off-chain transmissions.  

The \textsc{TempTranscript} procedure computes the simulated transcript 
by concatenating all the messages that according to the current knowledge (a) were received uncorrupted and (b)  were generated (at the sender's side)  according to ``correct'' information, i.e., information that is consistent with the current chains.
To clarify this behavior, 
consider round~$i$ where, without loss of generality, Alice sends the message~$m_i$.
Assume that 
the last transmission \emph{received} by Alice prior to round~$i$, which we denote~$\Prev(i)$ (see Definition~\ref{def:prev}), is uncorrupted. This implies that Alice learns in round~$\Prev(i)$ the \emph{correct} current chain of Bob, i.e., which of Bob's transmissions so far are correct and which ones are corrupted. Using the feedback, Alice knows which of her transmissions were corrupted and thus she knows both current chains. 
Learning the correct chains allows Alice to retrieve a correct partial transcript of~$\pi_0$ (Lemma~\ref{lem:uncorrupt-Goodtrans}). 
Hence, she can generate the correct~$m_i$ that extends the simulation of~$\pi_0$ by one symbol.

In each round of the protocol, the parties construct the partial simulated transcript implied by the current chains. If the received transmission is not corrupted, the $\Trans$ procedure retrieves the correct transcript (i.e., the simulated transcript at that point is indeed a prefix of the transcript of~$\pi_0$).
Then, the parties simulate the next rounds of~$\pi_0$ assuming the simulated partial transcript. As long as there is no noise in two alternating rounds, the next transmission extends the simulation of~$\pi_0$ by one symbol. Otherwise, the sent symbol may be wrong, however, it will be ignored in future rounds once the chains indicate that this transmission was generated due to false information.
Finally, at the end of the protocol, the parties output the simulated transcript implied by the 
\emph{longest} chain at each side. 
The main part of this section is the proof that the longest chain 
indeed implies a complete and correct simulation of~$\pi_0$.

An important property of the coding scheme is its adaptive order of speaking. In the first $2n/5$ rounds, the order of speaking alternates. In later rounds, the order of speaking is determined according to the observed noise: the more corrupted transmissions a party has, the \emph{less} the party gets to speak.
In particular, the protocol is split into \emph{epochs} of 2 or 3 rounds each. 
In the first two rounds of an epoch, the order is fixed: Alice speaks in the first round and Bob speaks in the second. 
Then, the parties estimate the noise each party suffered so far (namely, the length of their current chain) and decide whether or not the epoch has a third round and who gets to speak in that extra round. 
For Alice to be the speaker in the third epoch-round, her \emph{current} chain must be of length less than~$n/5$ while Bob's current chain must be longer than~$n/5$; Bob gets to speak if his chain is of length less than~$n/5$ while Alice's chain is longer than~$n/5$. 
In all other cases, the epoch contains only two rounds. 
We emphasize that due to the noiseless feedback, both parties agree on the received symbols (on both sides), which implies they agree on the current chains on both side, and thus, on the order of speaking in every epoch.
The $\Next$ procedure (line~\ref{alg:Next}), which determines the next speaker according to the current received transcript, captures the above idea.

The coding scheme is depicted in Algorithm~\ref{alg:codingLarge}.

\begin{algorithm}[htp!]
\caption{\textbf{(Part I)} A coding scheme against $(1/5,1/5)$-corruptions assuming noiseless feedback\\(Large alphabet; Alice's side)}
\label{alg:1515corruption}
\label{alg:coding}\label{alg:codingLarge}
\begin{algorithmic}[1]
\small
\Statex Input: A binary alternating protocol $\pi_0$ with feedback; noise parameter $1/5-\epsilon$.
Alice's input for~$\pi_0$ is~$x$.
\Statex Let $\Sigma=[n] \times \{0,1,\emptyset \}$. 

\Statex
\State 
\parbox[t]{\dimexpr\linewidth-\algmargin-4cm}{%
Throughout the protocol, maintain $S_A,R_A,R_B$, the sent, received by Alice, and received by Bob (as indicated by the feedback) symbols communicated up to the current round, respectively.}
\For {$i=1$ to $n=|\pi_0|/\eps$} 
	\State $p_{\text{next}}= \Call{Next}{R_A,R_B}$ \Comment{Determine the next party to speak}
	
	\If {$p_{\text{next}}=$Alice}
		\State $T\gets$\Call{TempTranscript}{$S_A,R_A,R_B$}
		\State \label{alg:NextSymbol}
		\parbox[t]{\dimexpr\linewidth-\algmargin-2cm}{%
		The next symbol $\sigma = (link,b)$ to be communicated is: \\
		\phantom{\hspace{\algorithmicindent}}$link $ is the latest \emph{non-corrupted} round $link<i$ where Alice is the speaker \\ \phantom{\hspace{\algorithmicindent}$link $} ($0$ if no such round exists). 
		\\
		\phantom{\hspace{\algorithmicindent}}$b = \pi_0(x\mid T)$ if Alice is the sender in $\pi_0$, otherwise (or if $\pi_0$ has completed) $b=\emptyset$.
		 }
	\Else
		\State (receive a symbol from Bob) 
	\EndIf
	\State \textbf{end if}
\EndFor
\State \textbf{end for}
\Statex
\Comment{Determine output by the longest chains}
\State $j \gets \argmax_i |\mathsf{Parse}({R_B^{\le i}})|$
\State $j' \gets \argmax_i |\mathsf{Parse}({R_A^{\le i}})|$
\State \textbf{Output} \Call{TempTranscript}{$S_A,R_A^{\le j'},R_B^{\le j}$}
\label{alg:last-step}

\end{algorithmic}
\end{algorithm}

\def\CommSpace{5cm}

\begin{algorithm}[pht!]
\caption{}   %
\caption*{\textbf{Algorithm~\ref*{alg:coding} (Part II)} The Parse, Next, and TempTranscript procedures}

\begin{algorithmic}[1]
\setcounterref{ALG@line}{alg:last-step}  %

\Statex
\Procedure{Parse}{$m_1,\ldots,m_t$} \label{alg:Parse}
\State $\textit{Chain} \gets \emptyset$, $j\gets t$
\While {$j>0$}
  \State $\textit{Chain}  \gets \textit{Chain}  \cup \{j\}$
  \State  $j \gets m_j.link$
\EndWhile
\State \Return \textit{Chain}
\EndProcedure 
\State \textbf{end procedure}

\Statex
\Procedure{TempTranscript}{$S_A,R_A,R_B$}   \Comment {\smash{\parbox[t]{\CommSpace}{Procedure for Alice; Bob's algorithm is symmetric}}} \label{alg:TempTrans}
	\State Set $G_B=\Parse(R_A)$
	\State \parbox[t]{0.5\columnwidth}{Set $G_A$ as all the rounds in which outgoing transmissions are not corrupted (as learnt by $R_B,S_A$)}
	\State For any $i<n$, if $\Prev(i),i \in G_A\cup G_B \cup \{0\}$ add $i$ to $\GChain$  \label{alg:trans:addgood}
	\State Set $T$ as the concatenation of all $\{b_i\}_{i\in \GChain}$, where $\sigma_i=(link_i,b_i)$ is the symbol received in round~$i$. 
	   \label{alg:trans:concatenate}
	\State \Return $T$
\EndProcedure
\State \textbf{end procedure}	%

\Statex

\Procedure{\Next}{$R_A,R_B$} \label{alg:Next}
	\State $i \gets |R_A|+|R_B|+1$  		\Comment {\smash{\parbox[t]{\CommSpace}{We are at round $i$}}}
	\State $j \gets 1$; $\Skip_A \gets 0$, $\Skip_B \gets 0$
		\Loop			
			\If {$i=j$} \Return Alice 
			\Comment {\smash{\parbox[t]{\CommSpace}{The speakers in the first two rounds of each epoch are fixed}}}
			\EndIf
			\If {$i=j+1$} \Return Bob \EndIf

			\Comment {\smash{\parbox[t]{\CommSpace}{Update the Skip counters}}}
			\If {$|\Parse(R_A^{< j+2})|\le n/5$} $\Skip_B \gets \Skip_B+1$ \EndIf \label{alg:step:SkipB}
			\If {$|\Parse(R_B^{< j+2})|\le n/5$} $\Skip_A \gets \Skip_A+1$ \EndIf \label{alg:step:SkipA}
			\Statex	
			\If  {$|\textsf{Parse}(R_B^{< j+2})| \le n/5 < |\textsf{Parse}(R_A^{< j+2})|$}
				\Comment {\smash{\parbox[t]{\CommSpace}{The epoch contains a 3rd round whenever one skip counter increases but the other does not}}}
				\If {$i=j+2$}  \Return Bob
				\Else \  $j \gets j+3$
				\EndIf
			\ElsIf {$|\textsf{Parse}(R_A^{< j+2})| \le n/5 < |\textsf{Parse}(R_B^{< j+2})|$}
								\If {$i=j+2$}  \Return Alice
				\Else\  {$j \gets j+3$}
				\EndIf
			\EndIf
			\State \textbf{otherwise} \State \hspace{\algorithmicindent} $j \gets j+2$						\Comment {\smash{\parbox[t]{\CommSpace}{An epoch with only 2 rounds}}}
		\EndLoop
		\State \textbf{end loop}
\EndProcedure
\State {\textbf{end procedure}}  

\Statex

\Statex \textbf{Note:} $R_A^{\le j}$ is the prefix of $R_A$ as received by the $j$-th round of the protocol (incl.~$j$) and $R_A^{< j}$ excluding round~$j$. The terms $R_B^{\le j}$ and $R_B^{<j}$ are similarly defined.

\end{algorithmic}
\end{algorithm}

\subsection{Basic properties}
Every transmission $m\in \Sigma=[n] \times \{0,1,\emptyset \}$ is interpreted as $m_i=(link_i,b_i)$ where $link_i$ points to a previous symbol~$m_j$.

\begin{definition}\label{def:prev}
For any $i\in [n]$
define $\Prev_A(i)$ as the maximal round $j<i$ where Alice is the speaker at round~$j$, or as~$0$ if no such $j$ exists.
Similarly, $\Prev_B(i)$ is the maximal round $j<i$ where Bob is the speaker at round~$j$, or as~$0$ if no such~$j$ exists. 
The notation~$\Prev(i)$ without a specific subscript refers to the maximal round~$j<i$ where the speaker in round~$j$ differs from the speaker of round~$i$.
\end{definition}

\subsubsection{Good rounds and the implied transcript}
Next, we show that whenever a transmission arrives correctly at the other side, the receiver learns all the uncorrupted transmissions communicated so far. First, let us define the notions of good rounds and implied transcript. 

\begin{definition}[Good Round]
A round $i$ is called \emph{good}  if the transmissions in both rounds~$i$ and $\Prev(i)$ are uncorrupted.
Let $\GOOD$ be all the good rounds and $\GOOD^{\le i} \triangleq \GOOD \cap [i]$.
\end{definition}
We additionally set any round $i \le 0$ to be good (and uncorrupted) by definition.

\begin{definition}[Implied Transcript]\label{def:trans}
For any round~$i$, the transcript $T(i)$ is defined as the  (natural order) concatenation of bits $\{b_j\}_{j\in \GOODi{i}}$ %
where $\sigma_j=(link_j,b_j)$ is the symbol transmitted (and correctly received) in round~$j$.
\end{definition}

While the above $\GOOD$ and $T(i)$ are tools for the analysis, 
the next lemma shows that whenever the $i$-th transmission arrives correctly at the other side, the receiver learns
$\GOODi{i}$ and $T(i)$. Specifically, the variable $\GChain$ (Line~\ref{alg:trans:addgood}) equals $\GOODi{i}$ and $\Trans$ outputs $T(i)$. This allows us (despite some abuse of notation) to treat $T(i)$ and the output of $\Trans$ interchangeably, as long as round~$i$ is uncorrupted.

\begin{lemma}\label{lem:uncorruptedMeansGood}
For any $i\in [n]$,
if the transmission of round~$i$ is uncorrupted, then  
$\GChain= \GOODi{i}$ at the receiver, and $\Trans$ outputs $T(i)$.
\end{lemma}
\begin{proof}
Assume, without loss of generality, that Alice is the receiver of the $i$-th transmission.
Since transmission~$i$ is uncorrupted, it holds that $G_B=\Parse(R_A)$ contains exactly all the uncorrupted transmission sent by Bob so far. Alice knows her own uncorrupted transmissions $G_A$ via the feedback. 
Then $\GChain$ indeed holds all the good rounds up to round~$i$, and $\Trans$ outputs $T(i)$ by definition.
\end{proof}
\begin{remark}
Assume that round~$i$ is corrupted and let $j<i$ be the latest uncorrupted round. 
Then $T(i) = T(j)$, which equals the output of $\Trans$ in round~$j$. 
However, the output of $\Trans$ in round~$i$ may be arbitrary.
\end{remark}

The next lemma argues that, if round~$i$ is uncorrupted, then the implied transcript~$T(i)$ (and hence the output of $\Trans$ in round~$i$) is indeed a correct (partial) simulation of~$\pi_0$.
\begin{lemma}\label{lem:uncorrupt-Goodtrans}
If round~$i$ is uncorrupted, then $T(i)$ is a prefix of~$\pi_0(x,y)$. 
\end{lemma}
\begin{proof}
The  proof goes by induction on~$i$. The base case $T(0)=\emptyset$ is trivial. Assume that the claim holds for $T(j)$ for any uncorrupted round~$j<i$; we show that the same holds for round~$i$.

Assume, without loss of generality, that Alice is the receiver in round~$i$. 
Let $j$ be the maximal previous round where Alice's transmission was not corrupted.
Since round $j$ is uncorrupted, Lemma~\ref{lem:uncorruptedMeansGood} proves that, at round~$j$, Bob learns $\GOODi{j}$ and $T(j)$. By the induction hypothesis, $T(j)$ is a prefix of $\pi_0(x,y)$. 

If $j<\Prev_A(i)$ then $i$ is \emph{not} a good round, $i\notin \GOODi{i}$. It holds that $\GOODi{i}=\GOODi{j}$ and $T(i)=T(j)$; therefore, $T(i)$ is indeed a prefix of $\pi_0(x,y)$. Otherwise, $j=\Prev_A(i)$ and $i$ \emph{is} a good round. 
As said, in round $j$ Bob learns $T(j)$ (which is a correct prefix of $\pi_0(x,y)$).
Next, in round~$i$ it is Bob's turn to send the symbol~$\sigma_i=(link_i,b_i)$. 
If it is Bob's turn to speak in $\pi_0$, then $b_i=\pi_0(y \mid T(j))$ will indeed be the correct continuation of~$T(j)$ according to $\pi_0$; otherwise, Bob sends $b_i=\emptyset$. In both cases, the channel does not corrupt $\sigma_i$, Alice learns $\GOODi{i}$ and the implied transcript she constructs equals $T(i)=T(j) \circ b_i$. Hence, $T(i)$ is indeed a prefix of~$\pi_0(x,y)$.
\end{proof}

\subsubsection{Skipped rounds, the order of speaking and noise-progress tradeoffs}
The order of speaking in the protocol depends  on the observed noise measured through the length of the current chain. Whenever the current chain is shorter than $n/5$ for only one of the parties, this party ``skips'' one round of communication---the other party gets to speak one additional round. We now define the skipping mechanism and use it to show that the coding scheme makes progress unless too much noise has occurred.  
\begin{definition}
We say that the epoch that starts at round~$i$ 
is \emph{Alice-skipped} if $|\Parse(R_B^{<i+2})| \le n/5$. Similarly, it is \emph{Bob-skipped} if $|\Parse(R_A^{<i+2})| \le n/5$. 
\end{definition}
Note that an epoch can be both Alice- and Bob-skipped.
Whenever an epoch is Alice-skipped, the counter $\Skip_A$ increases by one (Line~\ref{alg:step:SkipA}) and Alice speaks only one time in that epoch. Similarly, in a Bob-skipped epoch, $\Skip_B$ increases by one and Bob speaks only once.

Next we prove some properties with regard to the number of rounds each party gets to speak, as a function of the noise. In particular, we relate between the variables $\Skip_A,\Skip_B$ and the number of rounds Alice and Bob get to speak, denoted $\RC_A, \RC_B$, respectively.

\begin{lemma}\label{lem:Skip2RC}
Alice is the sender in $\frac{1}2(n-\mathsf{SkipCnt}_A+\mathsf{SkipCnt}_B)$ rounds and 
Bob is the sender
in $\frac{1}2(n-\mathsf{SkipCnt}_B+\mathsf{SkipCnt}_A)$ rounds.
\end{lemma}
\begin{proof}
We split the protocol into the epochs generated by the \Next procedure. 
For the $i$-th epoch denote $n(i)\in \{2,3\}$ the number of rounds in that epoch, and let $A(i)$ (resp., $B(i)$) be an indicator which is~$1$ if the epoch is Alice-skipped (resp., Bob-skipped).

Note that Alice speaks in the $i$-th epoch exactly $\frac12(n(i)-A(i)+B(i))$ times: 
if $n(i)=2$ it must hold that $A(i)=B(i)$ and Alice speaks once. She also speaks once if $n(i)=3$, but Bob speaks at the third round, $A(i)=1,B(i)=0$, i.e., if this is an Alice-skipped but not a Bob-skipped epoch. Finally, Alice speaks twice only when $n(i)=3$ and $A(i)=0,B(i)=1$. 
Then,
\[
RC_A = \sum_i \frac{n(i)-A(i)+B(i)}2= \frac{n - \Skip_A +\Skip_B}2.
\]
The case for Bob is symmetric.
\end{proof}

\begin{remark} In fact, due to rounding and the fact that Alice is the first to speak, she might get one extra round if the total number of rounds does not divide into full epochs, e.g., when the last epoch contains only a single round. 
A more accurate statement is $\RC_A \ge (n -\mathsf{SkipCnt}_A +\mathsf{SkipCnt}_B )/2 - 2$ (and similarly for Bob). In order to simplify the proof, we ignore this issue.
\end{remark}

Next, we connect the number of skips with the amount of noise that happens during the first part of the protocol.
\begin{claim}\label{clm:skip-bound}
If $t$ transmissions by Alice were corrupted during the $2n/5$ first rounds, then at the end of the protocol,
\[
\Skip_A \ge n/5 + t.
\]
\end{claim}
\begin{proof}
During the first $2n/5$ rounds, all the epochs are both Alice- and Bob-skipped (i.e., epochs of size 2).
This means that by round $i=2n/5$, $\Skip_A=n/5$. 

Split rounds $[2n/5+1,n]$ into epochs as done by the $\textsc{Next}$ procedure; note that there are at least $n/5$ epochs in this part of the protocol. Since the noise corrupted~$t$ of Alice's transmissions before round~$2n/5$, it can corrupt at most $n/5-\eps n-t$ additional transmissions of Alice beyond round~$2n/5$.

That is, in at least 
$n/5 - (n/5-\eps n-t)=t+\eps n$ of the epochs after round~$2n/5$, 
Alice's transmission (in the first round of the epoch) is not corrupted; call these epochs Alice-uncorrupted.
Note that by round $2n/5$, Alice's ``correct'' chain is of length at most $n/5-t$.
As long as the length of Alice's correct chain is less than~$n/5$, any Alice-uncorrupted epoch is also Alice-skipped. In each such epoch,  $\Skip_A$ increases by one and Alice gets to speak only once. The length of Alice's correct chain also increases by one in each such epoch.
It follows that in each of following $t$ Alice-uncorrupted epochs, $\Skip_A$ increases
until the length of the correct chain exceeds~$n/5$ and the condition of line~\ref{alg:step:SkipA} does not hold any longer.
may increase $\Skip_A$ even further). 
Since the number of Alice-uncorrupted epochs is $t+\eps n > t$, the counter will indeed reach at least~$n/5+t$.
\end{proof}

The following lemma captures a key property of our resilient protocol---a relation between the length of the implied transcript and the number of corruptions that have occurred so far.
\begin{lemma}\label{lem:stronger}
If, up to some round~$r$, there were $t$ Alice-skipped epochs where Alice's transmission is uncorrupted
and at most $t-k$ corruptions in Bob's transmissions, 
then 
\[
|T(r)| \ge k.
\] 
\end{lemma}
\begin{proof}
We prove the claim by induction on $k$ for all $r\ge t\ge k$.

The base case: $k=0$, trivially holds since for all~$r,t$, we have $|T(r)| \ge 0$.
For the inductive step, 
we assume that the lemma holds (for both parties) for some $k$ and any $r\ge t\ge k$, 
and wish to show that it also holds for~$k+1$ and any $r \ge t \ge k+1$.
Specifically, let $r,t$ where $r\ge t$ be fixed.
We are given that by round~$r$ 
there are $t$ Alice-skipped Alice-uncorrupted epochs, and that Bob's transmissions suffer from at most $t-(k+1)$ corruptions; we need to show that $|T(r)|\ge k+1$.

Since the number of corruptions at Bob's side by round $r$ is less than $t-k$, the induction hypothesis tells us that there exists some round $j\le r$ where $|T(j)|\ge k$. Let $j$ be the minimal round such that $j\in \GOOD$ and $|T(j)|=k$ while $|T(j-1)|=k-1$. 
Recall that~$T$ extends only in good rounds %
(Definition~\ref{def:trans}),
hence, such a round $j$ must exist.

Assume there are $t'$ Alice-skipped Alice-uncorrupted epochs until round $j-1$.  
It follows that the number of corruptions at Bob's side (up to round~$j-1$) 
is at least $t'-k+1$: 
if the number of corruptions is strictly less than $t'-k+1$ then by induction $T(j-1)\ge k$, contradicting the way we chose~$j$. 
It follows that the number of corruptions in Bob's transmissions for round $[j,r]$ is
at most $(t-k-1) - (t'-k+1) = t-t'-2$. 

We now split the analysis into different cases.
Assume Alice is the speaker in the~$j$-th round. We know that $j\in\GOOD$ and  
$|T(j)|>|T(j-1)|$.
This implies that round $(j-1)$ is uncorrupted and that Bob is the next to speak in~$\pi_0$ given $T(j)$, due to the alternating nature of~$\pi_0$. 
Note that Alice has at least $t-t'$ additional uncorrupted rounds within Alice-skipped epochs in rounds $[j,r]$ (note that Alice speaks in round~$j$, which is uncorrupted). 
In all these cases, either Bob speaks a single time immediately after Alice, or he speaks twice after Alice (Alice-skipped epoch). 
Since at most $t-t'-2$ of Bob's transmissions are corrupted, 
it follows that 
there must exist an uncorrupted round $j'\in[j+1,r]$ where Bob is the speaker and  $\Prev_A(j')$ is  uncorrupted, i.e., $j'\in\GOOD$.
This implies that Bob sends the correct symbol that extends $T$ in round~$j'$; thus,
$|T(j')|=|T(j)|+1=k+1$.  Since $|T(\cdot)|$ is non-decreasing, we have proved the claim.

The other case is when Bob is the speaker in the $j$-th round.
Again, $j\in\GOOD$, since $|T(j)|> |T(j-1)|$, thus,  round~$j$ itself is uncorrupted. 
Then, in $[j+1,r]$ Alice has $t-t'$ additional uncorrupted rounds (in Alice-skipped epochs) while at most $t-t'-2$ of Bob's transmissions are corrupted. 
Similar to the previous case, after each one of the aforementioned Alice-skipped Alice-uncorrupted rounds, Bob either speaks once or twice. 
If we consider the previous ($\Prev_B$) of these $t-t'$ rounds of Alice, we know that at most $t-t'-2$ of them can be corrupted (notice that round~$j$ itself belongs to Bob and is uncorrupted!).
This means that there must exist a round~$j'\in[j+1,r]$ where Alice is the speaker and $\Prev_B(j')$ is uncorrupted, i.e., $j'\in\GOOD$. Then, $|T(j')|=|T(j)|+1=k+1$, which completes this case.
\end{proof}

An immediate corollary of the above lemma shows that the ``correct'' chain of the  coding scheme  fully simulates~$\pi_0$. Indeed, the proof of Claim~\ref{clm:skip-bound} suggests that by the end of the coding scheme there were at least $n/5$ uncorrupted Alice-skipped rounds. Furthermore, the corruption on Bob's side is bounded to $(1/5-\eps)n$. Lemma~\ref{lem:stronger} then gives that 
\[
T(n) \ge     %
\eps n = |\pi_0|.
\]
Yet, we still need to prove that  the protocol outputs, in round~$n$, 
a chain that contains $T(n)$ as a prefix. 
In other words, we need to prove that the \emph{correct} chain is a prefix of the \emph{longest} chain. This is the goal of Section~\ref{sec:resilience-proof} below.

Another useful corollary of Lemma~\ref{lem:stronger} is the following lemma that measures the progress in the first $2n/5$ alternating rounds, as a function of the total amount of corruptions in that part.
\begin{lemma}\label{lem:progressWithCorruptions}
If by round $i \le 2n/5$ there were at most $i/2-k$ corruptions then
$|T(i)| \ge k$.
\end{lemma}
\begin{proof}
This is an immediate corollary of Lemma~\ref{lem:stronger}.
Note that all the $n/5$ epochs up to round~$2n/5$ are both Alice-skipped and Bob-skipped epochs, and that the order of speaking alternates. 
Assume that Alice has~$t$ uncorrupted rounds until round~$i$, then Bob's transmissions suffer from at most $(i/2-k)-(i/2-t) = t-k$ corruptions. Lemma~\ref{lem:stronger} gives that $T(i)\ge k$.
\end{proof}

\subsection{Correctness: Resilience to $(1/5-\eps,1/5-\eps)$-corruptions}\label{sec:resilience-proof}

In this part we prove the following theorem.

\begin{theorem}\label{thm:resilience-Large}
For any $\eps>0$ and any binary alternating protocol~$\pi_0$, 
Algorithm~\ref{alg:coding} correctly simulates~$\pi_0$ over a noisy channel with noiseless feedback and is resilient to any $(1/5-\eps,1/5-\eps)$-corruption.
\end{theorem}

In the following we implicitly assume the noise is a $(1/5-\eps,1/5-\eps)$-corruption. 
As mentioned earlier, Algorithm~\ref{alg:codingLarge} simulates the entire transcript of~$\pi_0$ correctly in its \emph{good} rounds. However, the parties cannot tell which rounds are good rounds. Instead, we show that the transcript implied by the longest chain (of each side) contains the entire transcript of~$\pi_0$ as its prefix. 

Before proving the theorem, let us set some notations.
Consider a complete instance of the coding scheme of Algorithm~\ref{alg:coding}.
Let $P_A$ be the longest chain of Alice's transmissions as seen by Bob at the end of the protocol.
Formally, $P_A = \Parse(R_B^{\le j^{max}})$ with $j^{max} = \argmax_j |\Parse(R_B^{\le j})|$.
Given a chain $P_A$ we differentiate between several types of Alice's rounds:
\begin{enumerate}
\item Uncorrupted rounds that are on $P_A$---we denote these rounds as the set~$NC_A$.
\item Corrupted rounds that are on $P_A$---we denote these rounds as the set $D$.
\item Corrupted rounds that are not on $P_A$---we denote these rounds as the set~$J$.
\end{enumerate}
In a similar way we can define~$P_B$ as the longest chain of Bob's transmissions as observed by Alice at the end of the protocol, and $NC_B$ as the set of uncorrupted rounds that are on~$P_B$, where Bob is the speaker.

\begin{proof}
Assume an instance of the algorithm with $(1/5-\eps,1/5-\eps)$-corruption. We prove that the algorithm outputs the transcript of~$\pi_0$ correctly.

Let $P_A$ be the longest chain of Alice at the end of the protocol. Then,
\begin{equation}\label{eqn:lem:longest-path}
|P_A| \ge \RC_A-(1/5-\eps)n.
\end{equation}
This holds since, in any uncorrupted round, Alice's transmission contains a link to the longest previous correct chain (Line~\ref{alg:NextSymbol}), thus extending this chain by at least one link. 
These uncorrupted rounds where Alice is the speaker form a chain of length at least $\RC_A-(1/5-\eps)n$. The longest chain at the end of the protocol, $P_A$, may only be longer. 

We can further classify each transmission of Alice, and understand its effect on~$P_A$, i.e., whether it belongs to the set $NC_A$ of uncorrupted rounds, the set $D$ of harmful corrupted rounds (that got into~$P_A$), or the set~$J$ of corrupted rounds that are not on~$P_A$, and thus are somewhat harmless.

It is easy to verify that
\begin{align} 
&|J|+|D| \le (1/5-\eps)n,	\label{eqn:NoiseLimit}\\
&|NC_A|+(1/5-\eps)n -|J| \ge \RC_A - (1/5-\eps)n \label{eqn:noiseCor}.
\end{align}
Eq.~\eqref{eqn:NoiseLimit} follows trivially from bounding the noise to a $(1/5-\eps,1/5-\eps)$-corruption.
Eq.~\eqref{eqn:noiseCor} is an immediate corollary of Eq.~\eqref{eqn:lem:longest-path} and Eq.~\eqref{eqn:NoiseLimit}, since $|P_A|=|NC_A|+|D|$.

Define $w_A,w_B$ to be the number of corrupted transmissions in the first $2n/5$ rounds (of Alice's and Bob's transmissions, respectively).
Furthermore we distinguish between before and after round~$2n/5$ via a prime and double prime superscripts, respectively, i.e.,
\begin{align*}
 &{NC}_A' = NC_A \cap [1, 2n/5],  & &{NC}_A'' = NC_A \cap [2n/5+1, n], \\
&D' = D \cap [1,2n/5], & &D'' = D \cap [2n/5+1,n], \\
&J' = J \cap [1,2n/5], & &J'' = J \cap [2n/5+1,n].
\end{align*}

\begin{lemma}\label{lem:NCA2WB}
If $|NC_A'| - w_B > \eps n$, then $P_A$
contains, as a prefix, a correct and complete simulation of~$\pi_0$.
\end{lemma}
\begin{proof}
Indeed, up to round $2n/5$ there were at  most $n/5-|NC_A'|$ corruptions on Alice's side and $w_B$ corruptions on Bob's, with a total of $n/5-|NC_A'|+w_B <n/5-\eps n$ corruptions.
Lemma~\ref{lem:progressWithCorruptions} implies that the progress up to round $2n/5$ is at least $\eps n$, that is, $|T(2n/5)| \ge \eps n = |\pi_0|$. 
Then, the entire transcript of~$\pi_0$ is correctly simulated by the $|NC_A'|$ uncorrupted rounds, and these rounds are on the chain~$P_A$.
\end{proof}

We now show that the output of Algorithm~\ref{alg:coding} is indeed correct.
Assume towards contradiction that the longest chain~$P_A$ does not imply the correct answer.
Then, Lemma~\ref{lem:NCA2WB} suggests that the number of corruptions in Bob's transmissions in the first $2n/5$ rounds is $w_B \ge |NC_A'|-\eps n$. Claim~\ref{clm:skip-bound} then implies that $\Skip_B \ge n/5+w_B \ge n/5 +|NC_A'|-\eps n$. Thus,
\begin{equation}
\Skip_B \ge |NC_A'| + (1/5 -\eps)n.	\label{eqn:skipB-large} 
\end{equation}

Further, 
\begin{equation}
\Skip_A \le 2n/5 -|NC_A'| +|J''|,		\label{eqn:skipA-small}
\end{equation}
since $\Skip_A$ increases by $n/5$ during the first $2n/5$~rounds, 
and at most by 1 in every round in~$J''$.
Rounds in $NC_A'' \cup D''$ can increase the counter only until the length of the chain reaches~$n/5$, that is, at most $n/5-|NC_A'|$ times. Putting these all together gives Eq.~\eqref{eqn:skipA-small}.

We show that in this case we have $|NC_A| + (1/5-\eps)n -|J| < \RC_A -(1/5-\eps)n$,
which contradicts Eq.~\eqref{eqn:noiseCor}.
Note that via Lemma~\ref{lem:Skip2RC}, the above can be written as
\begin{align}
|NC_A| - (1/10+2\eps)n  +\Skip_A/2 -\Skip_B/2 < |J|.
\label{eqn:contradiction}
\end{align}
The above bounds on the skip-counters, Eq.~\eqref{eqn:skipB-large} and Eq.~\eqref{eqn:skipA-small}, allow us to bound
the left-hand side of Eq.~\eqref{eqn:contradiction}  by
\begin{align*}
 &|NC_A| - (1/10+2\eps)n  +\Skip_A/2 -\Skip_B/2 \\
 &\le  |NC_A| - (1/10+2\eps)n + n/5 - |NC_A'|/2 +|J''|/2 - |NC_A'|/2 - (1/5 -\eps)n/2 \\
 &\le  |NC_A''| +|J''|/2 -3\eps n/2.  \numberthis\label{eqn:ref1}
\end{align*}

Now, if $|NC_A''| \le 3\eps n/2$, then Eq.~\eqref{eqn:contradiction} holds and we reached a contradiction. 
Otherwise, $|NC_A''| >0$ which means that $D\cap [2n/5]=\emptyset$, hence,
$|NC_A'|+|J'|=n/5$. %
Assume that $|NC_A''| \le n/5 - |NC_A'|$ (we will prove this shortly), then
  Eq.~\eqref{eqn:ref1} is upper bounded by 
\begin{align*}
&\le (n/5-|NC_A'|) + |J''|/2 - 3\eps n /2 \\
&\le |J'| + |J''|/2 - 3\eps n/2 \\
&< |J|.
\end{align*}
We obtained a contradiction for the second case as well. 
We are left to show that the assumption we took earlier holds, i.e., that
\[
|NC_A'| + |NC_A''| \le n/5.
\]
If the above equation does not hold, then there are $n/5$ Alice-skipped Alice-uncorrupted transmissions \emph{on the chain that becomes the output}. 
Since the number of Bob's corrupted transmissions is limited to $n/5-\eps n$, 
Lemma~\ref{lem:stronger} immediately gives that %
the length of the correct simulation of~$\pi_0$ is at least $\eps n = |\pi_0|$.
This transcript is contained in the output chain and contradicts the assumption that the longest chain implies an incorrect output.
\end{proof}

Finally, we argue that  Algorithm~\ref{alg:codingLarge} is computationally efficient, as long as $\pi_0$ itself is efficient.
\begin{proposition}\label{prop:efficiencyLarge}
For any constant $\eps>0$ and any $\pi_0$ given as a black-box, Algorithm~\ref{alg:codingLarge} is computationally efficient in~$|\pi_0|$.
\end{proposition}
\begin{proof}
The algorithm performs $n=|\pi_0|/\eps$ iterations, in each of which it needs to determine the next speaker, determine the partial transcript so far and determine the next message to send.
The former two activities require performing $\Parse$ on all the symbols received by both parties; this takes~$O(n)$ time. Setting the next message requires a single activation of~$\pi_0$.
\end{proof}

\section{A Coding Scheme with a Constant-Size Alphabet}
\label{sec:codingSmall}

\subsection{From large to constant alphabet:  Overview}
The coding scheme of Section~\ref{sec:largeAB}
uses an alphabet whose size is polynomial in~$n$, which is large enough to describe links to each of the $n$ rounds of the protocol.
We now show how to decrease the size of the alphabet to a constant.
The main, and quite natural, idea is to encode each link using several symbols. 
We will use a constant-size alphabet~$\Sigma$ of size $|\Sigma| \approx C^2$, where $C$ is some constant we set later as a function of~$\eps$, i.e., $C=O_\eps(1)$. 
We interpret each symbol $m\in\Sigma$ as the triplet $(link,type,msg)$, where
$link\in \{0,\dotsc,C\}$, $type \in \{ std, start, stop, cont \} $ and $msg\in [C]\cup \{0,1,\emptyset\}$. 

In order to link to a transmission which is at most $C$ transmissions back, 
the $link$ field can be used directly to contain a relative pointer.
That is, $link=1$ means the previous transmission, $link=2$ means the second previous transmissions, etc.
In this case, $type=std$ and the $msg$ field contains the payload---the bit $b\in \{0,1\}$ sent by the party according to $\pi_0$ (or $msg=\emptyset$ if the other party is to speak in~$\pi_0$).

When the protocol needs to link to a transmission which is $x>C$ transmissions back, we use a variable-length encoding of the relative pointer. 
Specifically, the coding begins with a message with 
$type=start$.
Next, 
the value of $x$ is encoded in the $msg$ fields of the next $\log_C x$ transmissions. 
In each such segment (except for the first one), the $link$ field still points to the last uncorrupted transmission. The $type$ field equals $cont$ to denote this transmission is a (middle) fragment of the encoding. 
On the last fragment, $type=stop$ denotes the end of the encoding.

A possible problem occurs when a party wishes to send an encoding of some (large) value $x$, but during the transmission 
of this encoding many corruptions occur. Due to the noise, the $link$ field of some specific segment of the encoding of~$x$ is too small to point to the previous segment. 
For example, say the two segments are $y>C$ transmissions apart. 
In this case, the above encoding acts recursively. 
That is, we initiate a new encoding (for $y$) by sending a message with $type=start$, whilst the encoding of~$x$ is still in progress. 
In the following transmissions, the $msg$ fields contain the value~$y$. 
After all the bits of $y$ have been transmitted, a message with $type=stop$ indicates the end of $y$'s encoding. Then, the encoding of $x$ resumes from the point it stopped. Once all the fragments of $x$ have been communicated, a message with $type=stop$ indicates the end of $x$'s~encoding, and the protocol continues as before.

This encoding does not harm the rate of the coding: 
most of the time the pointer is small enough and fits in a single $link$ field with no further encoding ($type=std$).
A burst of $t>C$ consecutive corruptions causes the addition of $\lceil \log_C t\rceil$ transmissions that describe a pointer to $t$ transmissions beforehand. 
It is not too difficult to verify that $n/C$ is a bound on the total added communication due to these encodings. We can set $C=1/\eps$ so that the added communication is bounded by $\eps n$~transmissions. 

These transmissions do not take part in the simulation of~$\pi_0$ and can be considered as a ``corruption'' towards that goal (although they serve a critical role in generating the uncorrupted chain). We argue that the effect of these transmissions on the simulation of~$\pi_0$ is at most as harmful as $\eps n$ corrupted transmissions. It then follows that if the noise corrupts at most $1/5-2\eps$ transmissions in each direction, the ``effective'' noise level (including transmissions used for encoding links) is bounded by $1/5-\eps$, which is low enough to allow the correct simulation of~$\pi_0$.

\subsection{A coding scheme with a constant-size alphabet}

Towards a scheme with constant-size alphabet let us (re)define some of the basic elements we use. 
Let $C=1/\eps$ be constant (without loss of generality, we assume $C$ is an integer). 
We define our alphabet to be
\[
\Sigma=\{0,\dotsc,C\} \times \{std,start,stop,cont) \times ([C] \cup \{0,1,\emptyset\}).
\]
Every $m\in \Sigma$ is interpreted as  
$m=(link,type,msg)$, where $link$ points to a previous symbol~$m'$ 
unless $type = start$, which indicates that the link to $m'$ is encoded in the $msg$ field of the next symbols. The $type$ field indicates whether the encoding has been completed ($type=stop$) or it is still going on ($type=cont$).
We emphasize that whenever $type \ne start$ the $link$ field indeed points to the previous uncorrupted transmission. We let $link=0$ indicate the first message in the chain (no previous message).

For any $m_1,\dotsc, m_t \in \Sigma$, the ``chain'' of messages,
$\Parse(m_1,\dotsc,  m_t)$, is determined by going over the chain link-by-link, until 
we hit the head of the chain ($link=0$) or an encoded link ($type\ne std$). In this case we collect the fragments of the link (recursively, in case we hit another instance of encoding before we are done collecting all the fragments of the  current encoding), decode them and continue parsing from the transmission pointed by the encoded value. The fragments that contain the encoding are omitted from the parsed output (so that the chain contains only the ``real'' messages of $\pi_0$). 
The $\Parse$ procedure is formally described in Algorithm~\ref{alg:Parese:smallAB}.

\begin{algorithm}[htp]
\caption{The $\Parse$ procedure for constant-size alphabet coding schemes}
\label{alg:Parese:smallAB}
\begin{algorithmic}[1]

\Procedure{Parse}{$m_1,\dotsc, m_t$}
\State $j \gets t$
\State $\Chain\gets \emptyset$

\While {$j$ monotonically decreases and $j > 0$}
\If {$m_j.type=std$} 
	\State $\Chain \gets \Chain \cup \{j\}$
	\State $j \gets j -m_j.link$
\Else
 	\State $j \gets j-{}$\Call{EffectiveAddress}{$m_1, \dotsm, m_j$}
\EndIf
\State {\textbf{end if}}
\EndWhile
\State {\textbf{end while}}
\State \Return $\Chain$
\EndProcedure
\State{\textbf{end procedure}}

\Statex
\Procedure{EffectiveAddress}{$m_1, \ldots, m_t$}
\If {$m_t.type \ne stop$}  \Return 0 \Comment{Error} \EndIf
\State $Temp \gets m_t$
\State  $j\gets t-m_t.link$
\While {$j$ monotonically decreases and $j > 0$}
	\If {$m_j.type = cont$} \Comment{Continue collecting fragments}
	\State	$Temp \gets Temp \cup \{m_j\}$
	\State	$j \gets j- m_j.link$
	\ElsIf {$m_j.type = start$} \Comment{All fragments are collected; decode msg fields}
			\State $Temp \gets Temp \cup \{m_j\}$
			\State  \Return \parbox[t]{0.6\columnwidth}{ the value obtained by concatenating the $msg$ fields in all the messages in $Temp$ in the natural order.\phantom{${}_a$}}
	\ElsIf {$m_j.type=stop$} \Comment{Recurse on inner encoding}
		\State $j\gets j-{}$\Call{EffectiveAddress}{$m_1,\ldots, m_j$}
	\ElsIf {$m_j.type =std$} \Comment{Should not happen whilst in encoding}
		 \State \Return 0
	\EndIf
	\State \textbf{end if}
\EndWhile
\State{\textbf{end while}}
\EndProcedure
\State{\textbf{end procedure}}
\end{algorithmic}
\end{algorithm}

The coding scheme with a constant-size alphabet is given in Algorithm~\ref{alg:codingSmall}.
It is very similar to the coding scheme of Algorithm~\ref{alg:codingLarge} except for the handling of 
encoded links, i.e., the encoding of a far link and parsing of a chain that contains encoded links.

\begin{algorithm}[htp]
\caption{A coding scheme with a constant-size alphabet  (Alice's side)}
\label{alg:codingSmall}
\begin{algorithmic}[1]
\small
\Statex Input: A binary alternating protocol $\pi_0$ defined over noiseless channels with feedback; a noise parameter $1/5-2\epsilon$. Alice's input for~$\pi_0$ is~$x$.
\Statex
\Statex Let $C=1/\eps$ and $\Sigma=\{0,\dotsc,C\} \times \{std,start,stop,cont) \times ([C] \cup \{0,1,\emptyset\})$. 
\Statex Without loss of generality, we assume $\log_2 C$ is an integer. 
\Statex The procedures \Call{Next}{}  and \Call{TempTranscript}{} are as described in Alrogithm~\ref{alg:coding}.

\Statex
\State Throughout the protocol, maintain $S_A,R_A,R_B$, the sent, received by Alice and received by Bob (as indicated by the feedback) symbols communicated up to the current round, respectively.
\State  $msgStack \gets \emptyset$
\For {$i=1$ to $n=|\pi_0|/\eps$} 
	\State $p_{\text{next}}= \Call{Next}{R_A,R_B}$ \Comment{Determine the next party to speak}
	
	\If {$p_{\text{next}}=$Alice}
		\State $T\gets$\Call{TempTranscript}{$S_A,R_A,R_B$}
		\State Let $\textit{lastMsg}$ be the offset to the latest \emph{uncorrupted} round where Alice is the speaker.   
		\State Let $b = \pi_0(x\mid T)$ if Alice is the sender in $\pi_0$, otherwise (or if $\pi_0$ has terminated) $b=\emptyset$.
		 \If {$lastMsg > C$} \Comment{Encode $link$ using multiple segments}
		 	\State Write $\textit{lastMsg}$ as a binary string $s=s_1 s_2\cdots s_t$ where $\forall i, |s_i|=\log C$
			\State $msgStack\gets push( (stop , s_t), (cont,s_{t-1}), \ldots, (cont, s_2), (start, s_1))$
		\EndIf
		\State {\textbf{end if}}
		\Statex 
		\If {$msgStack=\textsf{empty}$}   \Comment{Complete sending links before sending new messages}
			\State $link \gets \textit{lastMsg}$
			\State $type = std$
			\State $msg \gets  b$
		\Else
			\State $link \gets \textit{lastMsg}$ \Comment{Irrelevant if $type=start$, otherwise $\textit{lastMsg}\le C$}
			\State $(type,msg) \gets msgStack.pop()$
		\EndIf
		\State {\textbf{end if}}
		\State send the symbol $\sigma = (link,type,msg)$
	\Else	\Comment{Bob is the speaker}
		\State (receive a symbol from Bob) 
	\EndIf
	\State {\textbf{end if}}
\EndFor
\State \textbf{end for}

\Statex

\State $j \gets \argmax \mathsf{Parse}({R_B^{\le j}})$
\State $j' \gets \argmax \mathsf{Parse}({R_A^{\le j'}})$
\State \textbf{Output} \Call{TempTranscript}{$S_A,R_A^{\le j'},R_B^{\le j}$}
\end{algorithmic}
\end{algorithm}

Similar to Algorithm~\ref{alg:codingLarge}, the coding scheme of Algorithm~\ref{alg:codingSmall} is clearly computationally-efficient.
\begin{lemma}\label{lem:efficiencysmall}
For any constant $\eps>0$ and any $\pi_0$ given as a black-box, Algorithm~\ref{alg:codingSmall} is (computationally) efficient in~$|\pi_0|$.
\end{lemma}
The proof is similar to the proof of Proposition~\ref{prop:efficiencyLarge}, once it has been verified that the new $\Parse$ procedure still takes linear time in~$n$.

\subsection{Analysis}
\begin{lemma}\label{lem:EncBounded}
Let $E$ denote the set of all the rounds where the transmission is uncorrupted and has $type\ne std$ (i.e., is a part of an encoding).
Then
\[
|E| \le  \eps n\text{.}
\]
\end{lemma}
\begin{proof}
Any burst of $t>C$ corruptions causes at most $\lceil \log_C t\rceil$ uncorrupted transmissions with $type\ne std$, that encode a link to $t$ transmissions back. Due to the recursive manner of the encoding, a later burst of corruptions has no effect on the encoding of previous links, it only delays the rounds in which the first encoding is transmitted by the number of rounds needed to encode the link that comes after the later burst. In other words,  a burst of $t$ corruptions followed by a burst of $t'$ corruptions cause at most $\lceil \log_C t\rceil + \lceil\log_C t'\rceil$ rounds with $type\ne std$.
Since the total number of corrupted rounds (per party) is bounded by $(1/5-\eps)n$,  the total encodings length (for that party) is bounded by $n/C$. 

Partition $E$ into $E_A,E_B$,
the encoding rounds on Alice's and Bob's sides, respectively.
Assume that the noise pattern on Alice's transmission is composed of bursts of lengths $t_1,t_2,\ldots, t_k$ where for every $i$ we have $|t_i|>C$ 
(otherwise $t_i$ does not add any transmissions with $type\ne std$).
Note that the above requirement implies that $k < n/5C$.
\begin{align*}
|E_A| = \sum_{i=1}^k \lceil \log_C t_i\rceil  
& \le k+  k \log_C \left(\sum_{i=1}^k \frac{t_i}{k}\right) \le k+ k \log_C \left (\frac{n}{5k}\right),
\end{align*}
where the first inequality follows from Jensen's inequality. $E_B$ is bounded by the same value.
The above function monotonically increases in~$[0,n/5]$.
The number of messages with $type\ne std$ is then upper bounded by the value of the function at $k=n/5C$,
\begin{align*}
|E| &\le |E_A| + |E_B| \\
&\le 2\left(\frac{n}{5C} +  \frac{n}{5C}\log_C\left(\frac{n}{5 \frac{n}{5C}}\right ) \right) = \frac{4n}{5C}  \\
&< \eps n.
\end{align*}
\end{proof}

We now prove that Algorithm~\ref{alg:codingSmall} simulates $\pi_0$ correctly as long as the corruption level is below~$1/5$. The idea is to reduce Algorithm~\ref{alg:codingSmall} to Algorithm~\ref{alg:codingLarge}. This is done by considering fragments of encoding as ``corrupted'' transmissions of Algorithm~\ref{alg:codingLarge}, while still obtaining the correct link from these encoded transmissions.
Since the number of transmissions used for encodings is at most $\eps n$, they ``increase'' the effective noise level by this small amount, which is still tolerable for Algorithm~\ref{alg:codingLarge}.

\begin{theorem}\label{thm:resilience-Small}
Algorithm~\ref{alg:codingSmall} is resilient to any $(1/5-2\eps,1/5-2\eps)$-corruption.
\end{theorem}
\begin{proof}

Algorithm~\ref{alg:codingSmall} differs from Algorithm~\ref{alg:coding} in one main aspect---rounds in which $type\in\{start,cont,stop\}$. Other than those rounds, the two algorithms behave exactly the same: 
given a similar transcript $m_1,\ldots, m_t$ for which $type=std$, 
they both generate exactly the same partial transcript, the same next message, and the same next speaker. 

We can interpret any instance of Algorithm~\ref{alg:codingSmall} as an instance of Algorithm~\ref{alg:coding} in which transmissions with $type\ne std$ correspond to ``erased'' transmissions in Algorithm~\ref{alg:coding}:  transmissions whose ``link'' part is invalid (hence, the parsed chain is empty).
Formally, there exists a transformation that takes any transcript $ m= m_1,\ldots, m_n$ generated by  Algorithm~\ref{alg:codingSmall} on the input~$x,y$ assuming a $(1/5-2\eps, 1/5-2\eps)$-corruption,
and generates a transcript ${m'} = m'_1,\dotsc,m'_n$ such that 
\begin{enumerate}
\item $ {m'}$ is an instance of Algorithm~\ref{alg:codingLarge} on the input $x,y$ that suffers from a $(1/5-\eps, 1/5-\eps)$-corruption.
\item For any $i\in [n]$, the parsed chain in both algorithms is the same,
$\Parse^{\text{Alg.~\ref{alg:codingSmall}}}(m_1,\ldots, m_i) = \Parse^{\text{Alg.~\ref{alg:codingLarge}}}(m'_1,\ldots,m'_i)$.
\end{enumerate}
The transformation is as follows:  
if $m_i.type =std$ and $m_i.link$ points to a message $m_j$ with $m_j.type=std$, then
$m'_i.msg = m_i.msg$ and $m'_i.link = j$. 
If $m_j.type=stop$ then $m'_i.link = \textsc{EffectiveAddress}(m_1,\ldots, m_j)$.
Other cases are irrelevant ($m'_i$ will be attributed to a corruption).

That is, the transmissions contain (logically) the same messages and links except for transmissions that contain encoded-links in~$m$. These  correspond to corrupted transmissions in~${m'}$. However, in every round~$i$ 
where $m_i$ links to the end of an encoded link ($m_j$), we set the link in $m'_i$ to
 $\textsc{EffectiveAddress}(m_1,\ldots, m_j)$,  
i.e., to the last non-encoding uncorrupted transmission prior to~$m_i$.

Item~2 holds by induction. Assume that the claim holds for all rounds up to~$i$. Since both algorithms generate the same parsed chain, they make identical decisions regarding the order of speaking and the identity of the next speaker. 
If the $(i+1)$-th transmission in~${m}$ 
links to a transmission more than $C$ steps back,  or if $m_{i+1}\ne std$, then
$m'_{i+1}$ is assumed to be corrupted. In this case it holds that 
$\Parse(m_1,\dotsc,m_{i+1}) = \Parse(m'_1,\dotsc,m'_{i+1}) = \emptyset$.

Otherwise, the $(i+1)$-th transmission links to a transmission $m_j$ at most $C$ steps back, and $m_{i+1}.type=std$.  If $m_j.type = std$, then $m'_{i+1}.link=j$ and the claim holds.
If $m_j.type=stop$, then $m'_{i+1}.link$ points to the link encoded by 
$\textsc{EffectiveAddress}(m_1,\ldots, m_i)$. 
Since  $\Parse$ in Algorithm~\ref{alg:codingSmall} resolves the identity of the message prior to $m_{i+1}$ as the one pointed by $\textsc{EffectiveAddress}(m_1,\ldots, m_i)$, it outputs the same sequence as $\Parse(m'_1,\dotsc,m'_{i+1})$ does in Algorithm~\ref{alg:codingLarge}.

As a consequence of Item~2, the parsed chains, and hence the implied transcripts, are identical between the two instances for any~$i\in[n]$. 
Therefore,  for any round~$i$, the transmission generated by 
Algorithm~\ref{alg:codingLarge} given $m'_1,\dotsc,m'_{i-1}$ equals $m'_i$ defined by the above transformation,  except for two cases: when $m_i$ is corrupted and when $m_i$ is an encoding ($m_i.type\ne std$).
Lemma~\ref{lem:EncBounded} bounds the number of encoded transmissions by~$\eps n$. Hence, any instance with a $(1/5-2\eps,1/5-2\eps)$-corruption in Algorithm~\ref{alg:codingSmall} translates to an instance of Algorithm~\ref{alg:coding} with a $(1/5-\eps,1/5-\eps)$-corruption.

The correctness of the  Algorithm~\ref{alg:codingSmall} follows from the correctness of the  Algorithm~\ref{alg:codingLarge}. 
\end{proof}

\section{Applications for Circuits with Short-Circuit Noise}\label{sec:circuits}

In this section, we prove our main theorems (Theorems~\ref{thm:main} and~\ref{thm:converse-inf}).
We show that the KW-transformation between formulas and protocols (and vice versa) 
extends to the noisy setting in a manner that preserves noise-resilience. 
Applying the results from Sections \ref{sec:codingLarge}--\ref{sec:codingSmall} onto the realm of boolean formulas gives a construction that is resilient to an optimal level of noise, namely, a fraction of $(1/5-\eps)$ of short-circuit gates in any input-to-output path.
Additionally, the results of Section~\ref{sec:imp} imply that noise-resilience of $1/5$ is maximal for formulas (assuming a polynomial overhead).

In the following subsections, we show how to convert between formulas and protocols while preserving their noise-resilience.
If we start with a formula that is resilient to $(\alpha,\beta)$-corruptions, our transformation yields a protocol that is resilient to $(\alpha,\beta)$-corruptions (Proposition~\ref{prop:FtoP}). 
Moreover, given a protocol that is resilient to $(\alpha,\beta)$-corruptions, the transformation yields a formula that is resilient to a similar level of noise (Proposition~\ref{prop:PtoF}).

\subsection{Preliminaries}\label{sec:circuits:prelim}

\paragraph{Formulas}
A formula $F(z)$ over $n$-bit inputs $z\in \{0,1\}^n$ 
is a $k$-ary tree where each node is a $\{\AND, \OR\}$ gate with fan-in~$k$ and fan-out 1. (While our results apply to any~$k$, in this section we will usually assume $k=2$ for simplicity.)
Each leaf is a literal (either $z_i$ or $\neg z_i$). The value of a node $v$ given the input $z\in\{0,1\}$, denoted $v(z)\in\{0,1\}$, is computed in a recursive manner: the value of a leaf is the value of the literal (given the specific input~$z$); the value of an $\AND$ gate is the boolean AND of the values of its $k$ descendants, $v_0,\cdots,v_{k-1}$, that is $v(z) = v_0(z) \AND \dotsb \AND  v_{k-1}(z)$. The value of an OR gate is $v(z) = v_0(z) \OR \dotsb \OR v_{k-1}(z)$. The output of the formula on~$z$, $F(z)$, is the value of the root node. We say that $F$ computes the function $f:\{0,1\}^n\to\{0,1\}$ if for any $z\in\{0,1\}^n$ it holds that $F(z)=f(z)$.

The depth of a formula, denoted $\text{depth}(F)$, is the longest root-to-leaf path in it. The size of a formula, denoted~$|F|$, is the number of nodes it contains. We denote by $V_\AND$ the set of all the $\AND$ nodes, and by $V_\OR$ the set of all the $\OR$ nodes.

\paragraph{Karchmer-Wigderson Games}\label{sec:KW}
For any boolean function $f:\{0,1\}^n \to \{0,1\}$, the \emph{Karchmer-Wigderson game}
is the following interactive task. Alice is given an input $x \in f^{-1}(0)$ and Bob gets $y\in f^{-1}(1)$.
Their task is to find an index $i\in [n]$ such that $x_i\ne y_i$. We are guaranteed that such an index exists since $f(x)=0$ while $f(y)=1$. We denote the above task by~$\KW_f$.

Karchmer and Wigderson~\cite{KW90} proved the following relation between formulas and protocols.
\begin{theorem}[\cite{KW90}]
For any function $f:\{0,1\}^n \to \{0,1\}$, the depth of the optimal formula for $f$ equals the length of the optimal interactive protocol for $\KW_f$.
\end{theorem}
The above theorem is proven by showing a conversion between a formula for $f$ and a protocol for $\KW_f$, which we term the \emph{KW-transformation}. In this conversion, the formula-tree is converted into a protocol tree, where every $\AND$-gate becomes a node where Alice speaks and every $\OR$-gate becomes a node where Bob speaks. For a node $v$, the mapping $a_v: \{0,1\}^n \to \{0,1\}$ is set as follows.
For a given input $z$, consider the evaluation of the formula $F$ on~$z$. The node $v$ is an $\AND$ gate and we can write $v(z)=v_0(z) \AND v_1(z)$ where $v_0$ and $v_1$ are $v$'s left and right descendants, respectively. If $v_0(z)=0$ we set $a_v(z)=0$; otherwise we set $a_v(z)=1$. For an $\OR$ gate denote $v(z)=v_0(z) \OR v_1(z)$, and $b_v(z)=0$ if $v_0(z)=1$; otherwise $b_v(z)=1$. If the protocol reaches a leaf which is marked with the literal $z_i$ or $\neg z_i$, it outputs~$i$.  For technical reasons we will assume that the protocol outputs either~$z_i$ or $\neg z_i$ rather than just giving the index~$i$. Note that the literal always evaluates to the value of~$f$; In this work, a $\KW_f$ protocol must satisfy this additional requirement.

It is easy to verify that the following invariant holds: for every node $v$ reached by the protocol on some input $(x,y)\in f^{-1}(0)\times f^{-1}(1)$, it holds that $v(x)=0$ while $v(y)=1$. This holds for the root node by definition, and our selection of mappings $a_v, b_v$ maintains this property. Specifically, for an $\AND$-gate $v$ for which $v(x)=0$ it must hold that at least one of the gate's inputs is zero.
Indeed, the way we chose $a_v$ advances the protocol to a child node that evaluates to~0. Since $v(y)=1$, both children of~$v$ evaluate to $1$ on $y$; 
thus both descendants satisfy the invariant.  
The analysis for an $\OR$ gate is symmetric. It follows that once the protocol reaches a leaf (the literal $z_i$ or $\neg z_i$), that literal evaluates differently on $x$ and on~$y$, so $x_i\ne y_i$ as required. In particular, the literal evaluates to~$0$ on~$x$ and to~$1$ on~$y$.

The same reasoning allows us to convert a protocol for $\KW_f$ into a formula for~$f$: consider the  protocol tree and convert each (reachable) node where Alice speaks to an $\AND$-gate and each (reachable) node where Bob speaks to an $\OR$-gate. If the protocol outputs $z_i$ or $\neg z_i$ at some leaf, that literal is assigned to that leaf.

Proving that this conversion yields a formula for $f$ is by induction on the length of the protocol. If $|KW_f|=0$, then the protocol outputs (say) $z_i$ without communicating. It is clear that all inputs in the domain satisfy $x_i \ne y_i$, and that $x_i=0$ while $y_i=1$ (negate these values if the output of the protocol is~$\neg z_i$). For the induction step, assume, without loss of generality, that Alice is to speak first. For some partition~$X^0 \cup X^1 = f^{-1}(0)$, Alice sends 0 when $x\in X^0$ and otherwise she sends~$1$. 
By induction, the continuation of the protocol can be converted into formulas $F_0$ and~$F_1$ (corresponding to the cases where Alice sends 0 and 1, respectively), for which $F_0(x)=0$ when $x\in X^0$,  $F_1(x)=0$ when $x\in X^1$, and $F_0(y)=F_1(y)=1$ when $y\in f^{-1}(1)$. Taking $F=F_0 \AND F_1$ completes the proof. The other case, where Bob is to speak first, is symmetric. See~\cite{KW90} for further details about the KW-transformation from formulas to protocols and vice versa, and for the formal proofs.

\begin{remark}\label{rem:faninAlphabet}
In the above, formulas are assumed to have fan-in 2 and protocols are assumed to communicate bits. However, the same reasoning and conversion also applies for a more general case, where each $\AND$-gate and $\OR$-gate has fan-in $k$, and the protocol sends symbols from an alphabet of size~$|\Sigma|=k$. 

Furthermore, while our claims below are stated and proved assuming fan-in~2, all our claims apply to any arbitrary fan-in~$k$.
\end{remark}

\paragraph{Short-Circuit Noise}
Short circuit noise replaces the value of a specific node with the value of one of its descendants.
A noise pattern $E\in \{0,1,\dotsc,k-1,\noerr\}^{|V_\AND| \cup |V_\OR|}$ defines for each node whether it is short-circuited and to which input. Specifically, if for some node $v$, $E_v=\noerr$, then the gate is not corrupted and it behaves as defined above. Otherwise, the value of the node is the value of its $E_v$-th descendant, $v(z) = v_{E_v}(z)$. 
We denote by $F_E$ the formula with short circuit pattern $E$; we sometimes write $F$ for the formula with no short-circuit noise, i.e., with the noise pattern $E=\noerr^{|V_\AND| \cup |V_\OR|}$.

We say that a circuit is resilient to a noise pattern~$E$ if for any $z\in\{0,1\}^n$ it holds that $F(z)=F_E(z)$.
\begin{definition}
We say that $F$ is resilient to a $\delta$-fraction of noise if it is resilient to all noise patterns~$E$ in which the fraction of corrupted gates in any input-to-output path in $F$ is at most~$\delta$. 
\end{definition}
We can also be more precise and distinguish between noise in $\AND$-gates and $\OR$-gates. 
\begin{definition}
An $(\alpha,\beta)$-corruption of short-circuit errors, 
is a noise pattern on a formula $F$ of depth~$n$
that changes at most $\alpha n$ $\AND$-gates and at most $\beta n$ $\OR$-gates 
in any input-to-output path in~$F$.
\end{definition}
\begin{remark}
Note that an $(\alpha,\beta)$-corruption is defined with respect to the maximal depth~$n$ of the formula. If the formula has shorter paths, then $\alpha$ and $\beta$ no longer describe the fraction of corrupted gates in these paths. Hence, we will assume $F$'s underlying tree is a \emph{perfect} $k$-ary tree: every inner node has exactly $k$ children,  and all the leaves are of the same depth~$n$. We denote these as perfect formulas. It is easy to convert every formula of depth $n$ to be perfect without affecting its function.
\end{remark}
The following is immediately clear by definition.
\begin{claim}
If, for some $\delta>0$,  the  formula~$F$ is resilient to any $(\delta,\delta)$-corruption of short-circuit errors, then $F$ is also resilient to a $\delta$-fraction of noise.
\end{claim}

On the surface, the other direction does not necessarily hold:
$(\delta,\delta)$-corruption may corrupt up to a fraction $2\delta$ of the gates in each path, hence,
resilience to a $\delta$-fraction appears to be insufficient to resist all $(\delta,\delta)$-corruptions.
Nevertheless, we argue that these two notions are indeed equivalent. 
The reason for this is that a short-circuit in an $\AND$-gate can only turn the output from $0$ to $1$. A short-circuit in an $\OR$-gate can only turn the output from $1$ to $0$. 
Then, if a formula evaluates to~$1$ on some input, the output remains~1 regardless of any amount of short-circuited $\AND$-gates. If the output is~$0$, it remains so regardless of any number of short-circuited $\OR$-gates.
This observation was already made by Kalai et al.~\cite{KLR12}.
\begin{lemma}[{\cite[Claim~7]{KLR12}}]\label{lem:one-sided-noise}
Let~$F$ be a formula, $z$ an input and $E$ any error pattern. 
Let $E_{\AND}$ be the error pattern induced by $E$ on the $\AND$-gates alone (no errors on the $\OR$-gates);
Let $E_\OR$ be the error pattern induced by $E$ on the $\OR$-gates alone. 
It holds that if $F_{E_\AND}(z)=0$ then $F_{E}(z)=0$,  and if 
$F_{E_\OR}(z)=1$ then $F_{E}(z)=1$.
\end{lemma}

The above lemma then implies that resilience to a $\delta$-fraction of noise corresponds to resilience to the same fraction of noise in both types of gates.
\begin{lemma}\label{lem:frac-to-corruptions}
If, for some $\delta>0$, the perfect formula~$F$ is resilient to a fraction $\delta$ of short-circuit noise, then $F$ is also resilient to any $(\delta,\delta)$-corruption.
\end{lemma}
\begin{proof}
Assume $F$ has depth~$n$ and consider any inputs $x,y$ such that $F(x)=0$ and $F(y)=1$. 

Let~$E$ be an arbitrary $(\delta,\delta)$-corruption pattern. In particular, $E$ short-circuits up to $\delta n$ of the $\AND$-gates and additionally up to $\delta n$ of the $\OR$-gates in any input-to-output path. 
Let $E_{\AND}$ be the error pattern induced by $E$ on the $\AND$-gates alone and let $E_\OR$ be the error pattern induced by $E$ on the $\OR$-gates alone. 
Note that both the noise patterns $E_\OR$ and $E_\AND$ corrupt at most a $\delta$-fraction of the gates in each path.

Since $F$ is resilient to a $\delta$-fraction of noise, we have 
\begin{align}
F_{E_\OR}(x)=F_{E_\AND}(x)=0,  \label{eqn:Noisy0} \\
F_{E_\OR}(y)=F_{E_\AND}(y)=1. \label{eqn:Noisy1}
\end{align}

Lemma~\ref{lem:one-sided-noise} and Eq.~\eqref{eqn:Noisy0} 
then imply that $F_E(x)=0$. Similarly, the lemma and Eq.~\eqref{eqn:Noisy1} imply that 
$F_E(y)=1$. Since the above holds for an arbitrary $(\delta,\delta)$-corruption $E$ and for all inputs $x,y$, we get that $F$ is resilient to  $(\delta,\delta)$-corruptions.
\end{proof}

Following the mapping between formulas and protocols, 
the authors in~\cite{KLR12} made the observations
that a short-circuit error in a formula translates to channel noise in the equivalent KW protocol, assuming both parties learn the noise, i.e., assuming noiseless feedback. Specifically, the feedback allows \emph{both} parties to continue to the same node in the protocol tree, despite the noise. Thus, it is crucial in order to keep the parties synchronized.
We will sometimes abuse notation and identify a short-circuit noise pattern with a transmission noise pattern for  a formula~$F$ and a protocol~$\pi$ that share the same underlying tree structure. Furthermore, we will denote  the two different objects with the same identifier~$E$.

\subsection{From Formulas to Protocols}
\label{sec:FtoP}

In this part we describe a variant of the KW-transformation, which we call the \emph{resilient KW-transformation} from formulas to protocols. We prove that there is a way to chose the mappings $a_v(x),b_v(y)$ in the protocol tree in a way that preserves resilience, that is, if~$F$ is resilient against $(\alpha,\beta)$-corruptions, then the resulting interactive protocol will feature the same resilience.

Recall that in the standard KW-transformation for some formula $F$ that computes~$f$, 
for any $(x,y)\in f^{-1}(0)\times f^{-1}(1)$, an invariant that $v(x)=1$ and $v(y)=0$ holds for any node~$v$ reached by the protocol given by the transformation. This invariant is the key for the one-to-one correspondence between the formula and the protocol. Keeping this invariant in the noiseless case amounts to selecting the mapping $a_v(x)$ to be the child of~$v$ that evaluates to~0 on~$x$, and $b_v(1)$ to be the child of~$v$ that evaluates to~1 on~$y$.

The main observation is that, given any $(\alpha,\beta)$-resilient formula~$F$ (i.e., a formula that is resilient to any $(\alpha,\beta)$-corruption), we can choose the mapping $a_v(x)$ as the child of~$v$ that evaluates to~0 \emph{given any $(\alpha,\beta)$-corruption}: such a child always exists! Similarly, the mapping $b_v(y)$ is set to be the child of~$v$ that evaluates to~1 given any $(\alpha,\beta)$-corruption.
This allows us to maintain the invariant that $v(x)=0$ and $v(y)=1$ for any node~$v$ reached by the protocol, regardless of the possible noise pattern.
Keeping this invariant leads to proving that the protocol correctly computes~$\KW_f$ despite $(\alpha,\beta)$-corruptions.

We begin by introducing our variant of the KW-transformation that preserves resilience.
\begin{definition}[resilient KW-transformation]\label{def:resKW}
Given any $(\alpha,\beta)$-resilient formula~$F(z)$, the resilient KW-transformation of $F$ yields an interactive protocol $\pi$ defined as follows over the domain $F^{-1}(0)\times F^{-1}(1)$.

\begin{enumerate}%
\itemsep0em 

\item
The formula-tree is converted into a protocol tree, where every $\AND$-gate becomes a node where Alice speaks and every $\OR$-gate becomes a node where Bob speaks.

\item
Order the nodes in the protocol tree in a BFS order starting from the root, and determine the mappings associated with each node in that order (i.e., before setting the mapping of some node, set the mapping of all its ancestors). 

\item 
For any inner node~$v$, 
let $S_{(v,x,y)}$ be the set of noise patterns $E$ such that $E$ is a $(\alpha,\beta)$-corruption and such that an instance of~$\pi$ given the input $(x,y)$ and noise~$E$ causes the protocol to reach the node~$v$ (note that this process is well defined due to the BFS order). 

\item  \label{nKW:mapping}
If $v$ is an $\AND$-node, for any $x$, the mapping 
$a_v(x)$ maps to the child $w$ for which the subformula of~$F$ rooted at $w$ evaluates to $0$ on~$x$ for  \emph{all} noise patterns~$E\in \bigcup_{y'\in F^{-1}(1)} S_{(v,x,y')}$.
If $v$ is an $\OR$-node, then for any $y$, the map $b_v(y)$ maps to the child~$w$ for which the subformula of $F$ rooted at $w$ evaluates to~$1$ on~$y$ for all noise patterns $E\in \bigcup_{x'\in F^{-1}(0)} S_{(v,x',y)} $. 

\item
A leaf of~$F$ marked with the literal  $z_i$ or $\neg z_i$ becomes a leaf (output) of the protocol with the same literal. 
\end{enumerate}
\end{definition}
Note that the mappings $a_v(x),b_v(y)$ defined in item~\eqref{nKW:mapping}  
may be partial functions.
Specifically, if an $\AND$-node $v$ is not reachable given the input~$x$ with any $y$ and any valid noise, 
then,
definition of $a_v$ on that input~$x$ has no meaning.

Proposition~\ref{prop:OneForAllNoisyAB} guarantees that for any reachable node~$v$ we can always find a child $w$ that satisfies the condition of item~\eqref{nKW:mapping}. 
The proposition further proves that every node~$v$ reached by the constructed protocol~$\pi$ (assuming any valid noise) satisfies the invariant that $v(x)=0$ and $v(y)=1$. Similar to the noiseless KW-transformation, this invariant would imply that~$\pi$ correctly computes~$\KW_f$ in a resilient manner (Proposition~\ref{prop:FtoP}).

\begin{proposition}\label{prop:OneForAllNoisyAB}
Let $F(z)$ be an $(\alpha,\beta)$-resilient formula, and consider the resilient KW-transformation  of~$F$ (Definition~\ref{def:resKW}).
For any node~$v$ reached during the construction,  and for any $(x,y)\in F^{-1}(0) \times F^{-1}(1)$ such that the partial protocol~$\pi$ constructed thus far reaches~$v$ on~$x,y$ and some $(\alpha,\beta)$-corruption~$E$, the following holds.

\begin{enumerate}[label=(\alph*)]
\item \label{clm:property-a}
$v(x)=0$ and $v(y)=1$ in~$F_E$.

\item \label{clm:property-b}
Let $F_{0}$ and $F_{1}$ be the subformulas (of~$F$) rooted at the left and right child of~$v$, respectively. There is at least one subformula $G\in \{F_{0},F_{1}\}$ that satisfies
$G_E(x) =0$ for all noise patterns $E\in \bigcup_{y'\in F^{-1}(1)}  S_{(v,x,y')}$ (when $v$ is an $\AND$-node), or
$G_E(y)=1$ for all noise patterns $E\in \bigcup_{x'\in F^{-1}(0)} S_{(v,x',y)}$ (when $v$ is an $\OR$-node).
\end{enumerate}
\end{proposition}
\begin{proof}
Let us begin with property~\textit{\ref{clm:property-a}}.
The proof is by induction on the depth of~$v$ in~$F_E$. 
For the base case, when $v$ is the root, $v(x)=F(x)=0$ and $v(y)=F(y)=1$ in~$F_E$ since $E$ is an $(\alpha,\beta)$-corruption to which $F$ is resilient.

Now, let $v$ be an arbitrary  node and let~$w$ be its parent in~$F_E$; we denote by~$u$ the other child of~$w$.\footnote{If $w$ has only one child, the claim would trivially hold. Extension to $k$-ary trees are straightforward.} Property~\textit{\ref{clm:property-a}} holds for~$w$ by the induction hypothesis. Consider the case where~$w$ is an $\AND$-gate (the case of an $\OR$-gate is shown in a similar manner). 
There are two cases according to the noise associated with~$w$. 
If there is no noise at~$w$, $E_w = \noerr$, then for any input~$z$ it holds that $w(z) = v(z) \AND u(z)$.
Using the induction hypothesis, $w(y)=1$, and it must hold that $v(y)=u(y)=1$. 
Additionally, 
$w(x)=0$ therefore at least one of $v(x)$ and $u(x)$ must be~0. 
The protocol $\pi_E$ proceeds to $v$ only if $v(x)=0$ for \emph{all} noise patterns in $\bigcup_{y'\in F^{-1}(1)}S(v,x,y')$, and, in particular, $v(x)=0$ for the noise~$E$ which clearly belongs to~$S(v,x,y)$.
If the protocol does not proceed to~$v$, then it is not reachable for $(x,y),E$ and the statement holds vacously.
The other case is when there 
 is noise at~$w$. Then $\pi_E$ reaches $v$ only if $E_w$ directs to the child~$v$. In this case
$w(z) = v(z)$, and thus $v(x)=w(x)=0$ and $v(y)=w(y)=1$ by the induction hypothesis.

We continue to the proof of property~\textit{\ref{clm:property-b}}. 
The base case for $v$ being the root node is a simple special case of the proof given below for an arbitrary~$v$.

Let $v$ be given and assume that the claim holds for all nodes~$v'$ that come before~$v$ in the BFS ordering. Specifically, it holds for all the ancestors of~$v$. 
We show that the claim holds for~$v$ as well.
Consider the case where $v$ is an $\AND$ node  (the other case is similar).
Assume towards contradiction that the claim does not hold for~$v$.
That is, there are two  noise patterns $E_0,E_1\in  \bigcup_{y'\in F^{-1}(1)}  S_{(v,x,y')}$ such that
${(F_0)}_{E_0}(x)=1$ and ${(F_1)}_{E_1}(x)=1$. 

Define the noise pattern~$E^*$ (over the nodes of~$F$) in the following way. 
For any ancestor of~$v$, $E^*$ is defined as the $\OR$-\emph{minimal} noise-pattern between $E_0$ and~$E_1$, i.e., the one that induces the least noise on $\OR$-gates in the root-to-$v$ path.
Furthermore, for any $\AND$-gate~$u$ in the root-to-$v$ path, if either $E_0$ or $E_1$ contain no noise at~$u$, set $E^*$ to have no noise at that gate. Otherwise, both $E_0$ and $E_1$ have noise at~$u$ and since both reach~$v$, the noise must be the same; in this case $E^*$ contains the same noise for $u$ as $E_0$ and~$E_1$. 
For the nodes that belong to the subformula~$F_0$, the noise $E^*$ is identical to~$E_0$, and for nodes that belong to the subformula~$F_1$, $E^*$ is identical to~$E_1$. 
For all other nodes there is no noise in~$E^*$.

Clearly by this construction, $E^*$ is an $(\alpha,\beta)$-corruption, since compared to either $E_0$ or $E_1$, we only reduced the amount of corruptions in both $\AND$ and $\OR$ gates between the root and $v$ (and kept the same number of corruptions below~$v$). 
Furthermore, there must  exist some $y'$ such that $\pi_{E^*}(x,y')$ reaches~$v$: 
assume $E_0$ was the $\OR$-minimal pattern. Since $E_0\in  \bigcup_{y'\in F^{-1}(1)}  S_{(v,x,y')}$ there exists $y'$ for which $\pi_{E_0}(x,y')$ reaches~$v$. We argue that $\pi_{E^*}(x,y')$ also reaches~$v$. Indeed, in any $\OR$-gate $\pi_{E^*}(x,y')$ behaves exactly like $\pi_{E_0}(x,y')$ since the noise in both is identical. 
For any $\AND$-gate $u$, $E_0$ may have noise in~$u$ while $E_1$ (and thus, $E^*$) does not.
However, there exists~$y''$ such that $\pi_{E_1}(x,y'')$ reaches~$v$. Hence, it also reaches~$u$, and it also advances to the same child as $\pi_{E_0}(x,y')$ does when it reaches~$u$. 
Since $u$ is an $\AND$-gate, this decision depends only on~$x$. 
By the above we learn that if the protocol reaches~$u$ and there is no noise, it advances to the same child determined by the noise $E_0$ at~$u$.
Therefore,  $\pi_{E^*}(x,y')$ takes the same child of~$u$ as $\pi_{E_0}(x,y')$. 
It follows that $\pi_{E^*}(x,y')$ reaches~$v$, and $E^*\in  \bigcup_{y'\in F^{-1}(1)}  S_{(v,x,y')}$.

Additionally, in $F_{E^*}$, the node $v$ evaluates to~$1$ on~$x$, 
because 
$(F_0)_{E^*}(x)=(F_0)_{E_0}(x)=1$ and $(F_1)_{E^*}(x)=(F_1)_{E_1}(x)=1$. 
But this contradicts property~\textit{\ref{clm:property-a}}, asserted at the beginning of this proof,
that for any noise $E$ (and specifically for~$E^*$), any node~$v$ that is reachable by $\pi_{E^*}(x,y')$ must evaluate to~$0$ on~$x$.
Therefore, at least one of $F_0(x)$ and~$F_1(x)$ evaluates to~$0$ on all noise patterns within the scope.
\end{proof}

With the above we can show our main proposition for converting formulas to protocols in a noise-preserving way. 
\begin{proposition}\label{prop:FtoP}
Let $F$ be a perfect formula 
that computes the function~$f$ and is resilient to 
$(\alpha,\beta)$-corruption of short-circuit gates
in every input-to-output path. 
Then, the resilient KW-transformation yields an interactive protocol~$\pi$ over channels with feedback, that solves $\KW_f$ and is resilient to $(\alpha,\beta)$-corruptions.
\end{proposition}
\begin{proof}
Let $E$ be a given $(\alpha,\beta)$-corruption, 
and let $\pi_E$ be the protocol defined above for~$F$
assuming the transmission noise induced by~$E$. 
We claim that the protocol $\pi_E$, that is, the protocol~$\pi$ under the noise~$E$, computes $\KW_f$, which means that $\pi$ is an $(\alpha,\beta)$-resilient protocol for~$\KW_f$.

Say that on inputs $(x,y) \in F^{-1}(0) \times F^{-1}(1)$ the protocol terminates at a leaf $v$ marked with either $z_i$ or $\neg z_i$. By Proposition~\ref{prop:OneForAllNoisyAB} it holds that $v(x)=0$ while $v(y)=1$ in~$F_E$ (and thus in~$F$), which implies that $x_i\ne y_i$. Note that the literal evaluates to the output of the function as we additionally require from~$\KW_f$ protocols.
\end{proof}

The conversion from resilient formulas into resilient protocols in Proposition~\ref{prop:FtoP} implies an upper bound on the maximal resilience of formulas, and proves Theorem~\ref{thm:converse-inf}.
\begin{theorem}\label{thm:converse}
There exists a function $f:\{0,1\}^n \to Z$ such that no formula~$F$ that computes~$f$ with fan-in~$k$ and depth less than $r<\frac56\frac{n}{\log 2k}$ is resilient to a fraction of $1/5$ of short-circuit noise.  
\end{theorem}

\begin{proof}
For $z\in\{0,1\}^n$, let $par(z)= z_1 \oplus \dotsm \oplus z_n$ be the parity function. 

Let $F$ be a perfect formula that computes $par(z)$ with 
AND/OR gates of fan-in~$k$ and 
$depth(F)<\frac56\frac{n}{\log 2k}$. Assume that~$F$ is resilient to a fraction of~$1/5$ of short-circuit noise.
Lemma~\ref{lem:frac-to-corruptions} shows that $F$ is also resilient to $(1/5,1/5)$-corruptions of short-circuits. 
Then, using Proposition~\ref{prop:FtoP} we obtain an interactive protocol $\pi$ for $\KW_{par}$ of length $|\pi|= depth(F)<\frac56\frac{n}{\log 2k}$ that communicates symbols from an alphabet of size $|\Sigma|=k$, and is resilient to $(1/5,1/5)$-corruptions. 
This contradicts Theorem~\ref{thm:KWnotresilient}.
\end{proof}

Note that computing the parity of $n$ bits can be done with a formula of depth~$O(\log n)$. However, the above theorem shows that any resilient formula for the parity function will have an exponential blow-up in depth, and thus exponential blow-up in size.
\begin{corollary}
There is no coding scheme that converts any formula~$F$
of size $s$ into a formula~$F'$ of size $o(\exp(s))$, such that
$F'$ computes the same function as $F$ and is resilient to $1/5$-fraction of 
short-circuit gates on every input to output path.
\end{corollary}

\subsection{From Protocols to Formulas}
\label{sec:PtoF}

In this part we show the other direction of the resilient KW-transformation, namely, that 
a resilient protocol can be transformed into a resilient formula, with the same resilience level. 
This result is based on the result in~\cite{KLR12}, adapted to the setting of $(\alpha,\beta)$-corruptions (rather than resilience to $\delta$-fraction of noise). Then, we can use our coding scheme %
that is resilient against  $(1/5-\eps,1/5-\eps)$-corruptions in order to transform any formula~$F$ into a $(1/5-\eps,1/5-\eps)$-resilient version.

\begin{proposition}\label{prop:PtoF}
Let $\pi$ be a protocol that solves $\KW_f$ for some function $f$ and is resilient to $(\alpha,\beta)$-corruptions.
The KW-transformation on the reachable protocol tree of~$\pi$ yields a formula~$F$ that computes~$f$ and is resilient to any $(\alpha,\beta)$-corruption of short-circuit noise in any of its input-to-output paths.
\end{proposition}

The above proposition is, in fact, a reformulation of a result by Kalai, Lewko, and Rao~\cite{KLR12}, implied by Lemma~\ref{lem:one-sided-noise} and the following. 
\begin{lemma}[{\cite[Lemma~8]{KLR12}}]\label{lem:KLR}
Let $f$ be a boolean function, and let $\pi$ be a protocol with root $p_{\text{root}}$. 
Let $T \subset f^{-1}(0) \times ([k] \cup \{*\})^{V_A}$ and 
$U \subset  f^{-1}(1) \times ([k] \cup \{*\})^{V_B}$
be two nonempty sets such that the protocol $\pi$ solves $KW_f$ on every pair of input and noise in $T\times U$, and
assume that any vertex that is a descendent of $p_{\text{root}}$ can be reached using some input and noise from  $T \times U$.

Then there is a formula $F$ that is obtained by replacing every vertex where Alice speaks with an $\AND$ gate, every vertex where Bob speaks with an $\OR$ gate and every leaf with a literal, 
such that for every 
$(x,E_A) \in T$, $(y,E_B)\in U$ it holds that
$F_{E_\AND}(x)=0$ and 
$F_{E_\OR}(y)=1$,
where $E_\AND$ is $E_A$ on Alice's vertices and $\noerr$ on Bob's vertices,
and 
$E_\OR$ is $\noerr$ on Alice's vertices and $E_B$ on Bob's vertices.
\end{lemma}

Using our coding scheme that is resilient to $(1/5-\eps,1/5-\eps)$-corruptions (Algorithm~\ref{alg:codingSmall})
we get that we can fortify any formula $F$ so it becomes resilient to a $(1/5-\eps)$-fraction of short-circuit noise, with only polynomial growth in size.
\begin{theorem}\label{thm:resF}
For any $\eps>0$, any formula $F$ of depth $n$ and fan-in $2$ that computes a function~$f$ can be efficiently converted into a formula~$F'$ that computes~$f$ 
even if up to a fraction of $1/5-\eps$ of the gates in any of its input-to-output paths are short-circuited. $F'$ has a constant fan-in $O_\eps(1)$ and depth $O(n/\eps)$.
\end{theorem}

\begin{proof}
The conversion is done in the following manner. 
Given a formula~$F$ (that computes some function~$f$) we first balance it, i.e., convert it to an equivalent formula $\tilde F$ of depth $\log |F|$ with no redundant branches. It is well known that such a formula always exists.
Next, we convert $\tilde F$
into a protocol $\pi$ for $\KW_f$ via the KW-transformation (Section~\ref{sec:KW}); note that the length of~$\pi$ is at most the depth of~$\tilde F$, that is, $O(\log |F|)$. Then, we convert $\pi$ into a protocol $\pi'$ that solves the same function $\KW_f$ and is resilient to $(1/5-\eps,1/5-\eps)$-corruptions, assuming noiseless feedback. This step is possible due to Theorem~\ref{thm:resilience-Small}.
The resilient $\pi'$ is then transformed back into the resilient formula~$F'$  that satisfies the theorem assertions, using Proposition~\ref{prop:PtoF}. 
Recall that the depth of the obtained formula is exactly the length of the resilient protocol. 

To complete the proof we only need to argue that the conversion can be done efficiently.
It is easy to verify that converting $\tilde F$ to $\pi$ is efficient, and also converting $\pi$ to $\pi'$ (Algorithm~\ref{alg:codingSmall}) is efficient by Lemma~\ref{lem:efficiencysmall}. 
The only part which is possibly inefficient is the reverse KW-transformation from~$\pi'$ back to a formula, which requires finding the reachable protocol tree of $\pi'$---the vertices~$v$ for which there exist an input~$(x,y)$ and a $(1/5-\eps,1/5-\eps)$-corruption~$E$ such that $\pi'(x,y)$ reaches~$v$ if the noise is~$E$. 
This part can be shown to be efficient by a technique similar to that presented in~\cite{KLR12}.
In Appendix~\ref{app:efficiency} we give a detailed proof.
\end{proof}

Theorem~\ref{thm:main} is an immediate corollary of the above theorem, by noting that 
\[
|F'| \le k^{\text{depth}(F')} 
=\left(2^{O(\log(1/\eps))}\right)^{O((\log |F|)/\eps)}
=\poly_\eps(|F|). 
\]
Here, $k\approx \eps^{-2}$ %
is the fan-in of $|F'|$ given by the alphabet size of the resilient interactive protocol~$\pi'$ constructed earlier.

\section*{Acknowledgments}
The authors would like to thank
Raghuvansh Saxena
and
the anonymous reviewers
for spotting an error in a preliminary version of this manuscript, and for numerous suggestions that greatly increased the readability of this  manuscript.

Mark Braverman is supported in part by NSF Award CCF-1525342 and the NSF
Alan T. Waterman Award, Grant No.\@ 1933331, a Packard Fellowship in
Science and Engineering, and the Simons Collaboration on Algorithms
and Geometry.  
Klim Efremenko is supported by the Israel Science Foundation (ISF) through Grant No.\@ 1456/18 and the European Research Council Grant No.\@ 949707.
Ran Gelles is supported in part by the Israel Science Foundation (ISF) through Grant No.\@ 1078/17 and the United States-Israel Binational Science Foundation (BSF) through Grant No.\@ 2020277.

%

%

\newcommand{\etalchar}[1]{$^{#1}$}

\appendix
\section*{Appendix}

\section{Theorem~\ref{thm:resF}: Efficiency}
\label{app:efficiency}
We now argue that the conversion from $\pi'$ to $F'$ in Theorem~\ref{thm:resF} can be done efficiently in the size of the formula~$F$.
Note that, in general, the conversion of Proposition~\ref{prop:PtoF} may not be efficient, but it is efficient for protocols obtained by the conversion described in Theorem~\ref{thm:resF}. The ideas in this part resemble the analysis presented in~\cite{KLR12} (yet in a somewhat more intuitive manner), and we sketch here the details for self completeness.

Theorem~\ref{thm:resF} requires us to find, in an efficient way, the reachable protocol tree of $\pi'$ assuming $(1/5-\eps,1/5-\eps)$-corruptions.
We show a slightly stronger claim: 
\begin{quote}\it
For any (efficiently computable\footnote{Specifically, given a noise pattern $E$, determining whether or not $E\in \Phi$ should be done efficiently.}) set of noise patterns~$\Phi$, we can obtain the $\Phi$-reachable protocol tree of~$\pi'$ in an efficient way.
\end{quote}  
Given a node $v$ of depth $h\le |\pi'|$ in the protocol tree of~$\pi'$, 
we say that $v$ is $\Phi$-reachable if there exist an input~$(x,y)$ and a noise pattern $E \in \Phi$ such that $\pi(x,y)$ reaches~$v$ when the noise is~$E$. The $\Phi$-reachable protocol tree of~$\pi'$ are all the nodes of depth at most~$|\pi'|$ that are $\Phi$-reachable.

Recall that $|\pi|=\log |F|$ and that $|\pi'|=O_\eps(|\pi|)$; set $d=|\pi'|$. In the following, ``efficiently'' means a  time-complexity of $\textrm{poly}(2^d) = \textrm{poly}(|F|)$.
The idea behind the algorithm is as follows. Given $v$ let $(v_0,v_1,...,v_h=v)$ be the nodes on the unique path from the root to~$v$. 
We examine each one of the possible noise patterns that affects only this path. That is, for each node $v_i$ we decide whether it is corrupted or not; there are at most $2^d$ different such noise patterns. 
For any fixed noise pattern, we verify that all the other edges ($v_i,v_{i+1}$) are consistent with the behavior of the simulation, and reject the noise pattern if they are not. If no inconsistency is found, we show that a valid run of $\pi'$ on some input with that noise pattern leads to~$v$.

First, we recall that we assume that the complete protocol tree of $\pi$ of depth~$|\pi|$ is reachable for some input, given there is no noise at all; that is, we prune all the redundant branches.\footnote{This assumption means that the formula $F$ we start with is optimal; note that obtaining the optimal formula $F$ for a given function may not be efficient. 
Yet, this assumption is not crucial. Alternatively (as performed in~\cite{KLR12}), we can assume that each leaf in $F$ is an independent variable---surely a resilient coding for such a formula would also be a resilient version of $F$, i.e., when only considering inputs that are consistent with~$F$. This, however, causes the (reachable) protocol tree of~$\pi$ to be larger, and respectively increases the size of the output resilient formula, yet keeping its size polynomial in~$|F|$. } 
\begin{assumption}\label{asm:Preachable}
For any node $v$ in the protocol tree of $\pi$ there exists an input $(x,y)$ such that an instance of $\pi$ on $(x,y)$ reaches~$v$. 
\end{assumption}

The $\Phi$-reachability test of~$v$ is performed by Algorithm~\textsc{Reach($\Phi,v$)} depicted below.
\begin{figure}[htb]
\begin{framed}
\noindent 
\textbf{Algorithm}~\textsc{Reach($\Phi,v$)}\textbf{:  $\Phi$-reachability check for $\pi'$}\\
Input: A set $\Phi$ of valid noise patterns; a node $v$ in a complete $k$-ary tree.  
\begin{enumerate}
\item  Given $v$ let $\gamma=(v_0,v_1,...,v_h=v)$ be the path from root to~$v$. 

\item Let $\phi_v$ be the set of all the noise patterns
$(E_{v_0},\ldots, E_{v_{h-1}})\in \{0,\dotsc, k-1,\noerr\}^{h-1}$ that affect the path~$\gamma$,
for which \\
(i) if $E_{v_i}\ne\noerr$, for some $0\le i< h$ then the ${E_{v_i}}$-th child of $v_i$ is $v_{i+1}$, and \\
(ii) $E\in \Phi$.

\item For all $E\in \phi_v$ repeat: 

Check if the path $\gamma$ and the noise $E$ are consistent with some input~$(x,y)$:
	\begin{enumerate}
	\item Verify that, as long there is no noise, $\gamma$ is consistent with the behavior of Algorithm~\ref{alg:codingSmall}, i.e., that the $link$ field in each message links to the previous uncorrupted message (which is known since the noise is known), that large links are encoded correctly (meaning, that $type$ is correct, and $msg$ corresponds to the correct pointer), etc.
	\label{alg:reach:ver1}
	\item 
	Loop over all leaves $l$ in~$\pi$, and let $\gamma_l$ be the path from the root to~$l$. 
	Check that $\gamma$ is consistent with~$\gamma_l$: as long as there is no noise the $msg$ field in messages with $type=std$ are indeed the ones implied by $\gamma_l$---they are consistent with the correct transcript of $\pi$ given that its input leads to the leaf~$l$.
	\label{alg:reach:ver2}
	\item If all verifications pass for a certain leaf $l$, output \texttt{Reachable}. 
	\end{enumerate}
\item output \texttt{Non-Reachable}.
\end{enumerate}
\end{framed}
\end{figure}

\goodbreak

\begin{claim}
Given an efficiently-computable set~$\Phi$,
Algorithm~\textsc{Reach($\Phi,v$)} takes time $\textup{poly}(2^d)$.
\end{claim}
\begin{proof}
It is easy to see that there are at most $2^d$ valid noise pattens for the path $\gamma$ that should be considered. For each, we need to go over all possible leaves in~$\pi$ and perform $O(d)$ checks per leaf; the number of such leaves is upper bounded by~$2^d$. The total time is clearly $\poly\;(2^d)$.
\end{proof}

\begin{theorem}
For any input node~$v$ and set~$\Phi$, Algorithm~\textsc{Reach($\Phi,v$)} outputs \textup{\texttt{Reachable}} 
if and only if there exists some input  $(x,y)\in F^{-1}(0)\times F^{-1}(1)$
and some noise $E\in \Phi$, such that $\pi'(x,y)$ reaches~$v$ when the noise is described by~$E$. 
\end{theorem}
\begin{proof}
It is easy to see that if an inconsistency is found for some noise~$E$ at Step~(\ref{alg:reach:ver1}), then the obtained $\gamma$ cannot describe a valid instance of~$\pi'$ with the noise~$E$ (regardless of the input). If the inconsistency is found at Step~(\ref{alg:reach:ver2}) it means that $\gamma$ cannot describe a valid instance of~$\pi'$ with noise $E$ and any input that leads to the leaf~$l$. 
Thus, if Step~(\ref{alg:reach:ver2}) fails for all leaves~$l$, there is no input that leads to~$v$ given that specific noise pattern~$E$.

It remains to show that if no inconsistency is found in steps (\ref{alg:reach:ver1})--(\ref{alg:reach:ver2}), then the node~$v$ is $\Phi$-reachable. This follows the same reasoning. 
Let $l$  be the leaf in~$\pi$ and~$E$ an error noise pattern for~$\pi'$
such that when Algorithm~$\textsc{Reach}$ checks $l,E$ it outputs that $v$ is reachable. Since $l$ is reachable in $\pi$ (Assumption~\ref{asm:Preachable}), let $(x,y)\in F^{-1}(0)\times F^{-1}(1)$ be an input that leads $\pi$ to the leaf~$l$; note that $(x,y)$ is a valid input for~$\pi'$. It is easy to verify that running $\pi'$ on $(x,y)$ with the noise $E$ yields exactly the path~$\gamma$. Therefore, $v$ is reachable.
\end{proof}

\end{document}